\documentclass[11pt]{article}
\usepackage[utf8]{inputenc}
\usepackage{pdfpages}
\usepackage[utf8]{inputenc}

\usepackage{graphicx}
\usepackage{scrextend}
\usepackage{amsthm}
\usepackage{comment}
\usepackage{graphicx,marvosym}
\usepackage{float}
\usepackage{tocbibind}
\usepackage{hyperref}
\hypersetup{
    colorlinks=false
}
\usepackage[english]{babel}
\usepackage{amsthm}

\usepackage[colorinlistoftodos,textsize=tiny,textwidth=2cm,color=red!25!white,obeyFinal]{todonotes}

\usepackage{geometry}
\usepackage{graphicx,adjustbox}
\usepackage{amsmath}
\usepackage{amssymb}
\usepackage[all]{xy}
\usepackage{tikz}
\usepackage[utf8]{inputenc}
\usepackage{csquotes}
\usepackage{cite}
\usepackage{lipsum} % for filler text
\usepackage[titletoc,toc,title]{appendix}
\usepackage{setspace}
\usepackage{arydshln}
\usepackage{xcolor}
\usepackage{graphicx}
\usepackage{subcaption, caption}
\usepackage{multirow}
\usepackage{tikz}  %for drawing rectangles in the matrix 
\usetikzlibrary{arrows,matrix,positioning}
\usetikzlibrary{matrix,fit,calc,positioning,arrows,chains,scopes,decorations.pathreplacing}
\newtheorem{theorem}{Theorem}[section]
\newtheorem{corollary}[theorem]{Corollary}
\newtheorem{lemma}[theorem]{Lemma}
\newtheorem{proposition}[theorem]{Proposition}

\newtheorem{definition}[theorem]{Definition}
\newtheorem{remark}[theorem]{Remark}

\newtheorem{claim}[theorem]{Claim}

\makeatletter
\renewcommand*\env@matrix[1][*\c@MaxMatrixCols c]{%
  \hskip -\arraycolsep
  \let\@ifnextchar\new@ifnextchar
  \array{#1}}
\makeatother

\DeclareMathOperator{\Ima}{Im}

\usepackage{csquotes}
\DeclareMathOperator*{\argmax}{argmax}
\newcommand{\cupdot}{\mathbin{\mathaccent\cdot\cup}}

\begin{document}

%\includepdf{mscthesis_titlepage.pdf}
%\thispagestyle{empty}

\title {The Communication Complexity of Payment Computation} 
\author{Shahar Dobzinski \and Shiri Ron\thanks{Weizmann Institute of Science. Emails: {\tt \{shahar.dobzinski, shiriron\}@weizmann.ac.il}. Work supported by BSF grant 2016192 and ISF grant 2185/19.}}

\maketitle

\begin{abstract}
	Let $(f,P)$ be an incentive compatible mechanism where $f$ is the social choice function and $P$ is the payment function. In many important settings, $f$ uniquely determines $P$ (up to a constant) and therefore a common approach is to focus on the design of $f$ and neglect the role of the payment function.
	
	Fadel and Segal [JET, 2009] question this approach by taking the lenses of communication complexity: can it be that the communication complexity of an incentive compatible mechanism that implements $f$ (that is, computes both the output and the payments) is much larger than the communication complexity of computing the output? I.e., can it be that $cc_{IC}(f)>>cc(f)$?
	
	Fadel and Segal show that for every $f$, $cc_{IC}(f)\leq exp(cc(f))$. They also show that fully computing the incentive compatible mechanism is strictly harder than computing only the output: there exists a social choice function $f$ such that $cc_{IC}(f)=cc(f)+1$. In a follow-up work, Babaioff, Blumrosen, Naor, and Schapira [EC'08] provide a social choice function $f$ such that $cc_{IC}(f)=\Theta(n\cdot cc(f))$, where $n$ is the number of players. The question of whether the exponential upper bound of Fadel and Segal is tight remained wide open.
	
	In this paper we solve this question by explicitly providing an $f$ such that $cc_{IC}(f)= exp(cc(f))$. In fact, we establish this via two very different proofs. 
	
	In contrast, we show that if the players are risk-neutral and we can compromise on a randomized truthful-in-expectation implementation (and not on deterministic ex-post implementation) gives that $cc_{TIE}(f)=poly(n,cc(f))$ for every function $f$, as long as the domain of $f$ is single parameter or a convex multi-parameter domain. We also provide efficient algorithms for deterministic computation of payments in several important domains. 
\end{abstract} 
\thispagestyle{empty}

%\newpage
%
%\section*{Acknowledgements}
%Various people have supported me throughout the process of working on the dissertation. I would like to thank my advisor Prof. Shahar Dobzinski, whose guidance, insight and curiosity contributed significantly to this research. Thank you, I learned a great deal from you.    
%I  also wish to thank my family for their moral support. Finally, I want to thank my partner Adar, for standing by me. 

%\thispagestyle{empty}
%\newpage
%
%\setcounter{page}{1}
%
%\tableofcontents
%\newpage

\section{Introduction}   

In a mechanism design problem we have $n$ players and a set $\mathcal{A}$ of alternatives. Each player $i$ has a valuation function $v_i:\mathcal{A}\rightarrow \mathbb R$ that specifies his value for each alternative. We assume that each $v_i$ belongs to some known set $V_i$. A basic question in mechanism design asks: given a social choice function $f\colon V_1\times\cdots\times V_n  \rightarrow \mathcal{A}$, are there payment functions $P_1,\ldots, P_n\colon V_1\times \cdots \times V_n\rightarrow \mathbb R$ that make $f$ incentive compatible? For now, we interpret incentive compatibility as truthful, ex-post implementation of $f$, that is: $P_1,\ldots,P_n$ satisfy that for every player $i$, $v_i,v'_i\in V_i$, and $v_{-i}$ that specifies the values of the other players, $v_i(f(v_i,v_{-i}))-P_i(v_i,v_{-i}) \geq v_i(f(v'_i,v_{-i}))-P_i(v'_i,v_{-i})$.

A highly successful paradigm in mechanism design is the ``prices do not matter'' paradigm. One pillar of this approach are various characterization theorems that provide relatively simple conditions for the implementability of social choice functions. Examples for such conditions include cycle monotonicity for all functions \cite{rochet1987necessary}, monotonicity for functions in ``single parameter'' domains \cite{archer2001truthful,myerson1981optimal}, and weak monotonicity for ``rich enough'' multi-parameter domains \cite{bikhchandani}. Another pillar are ``uniqueness of payments'' or ``revenue equivalence'' theorems. Those theorems state that in most domains if $P_1,\ldots, P_n$ and $P'_1, \ldots, P'_n$ are two possible payment functions for $f$, then for each $i$ and $v_{-i}$ there exists a constant $c$ such that $P_i(\cdot , v_{-i})-P'_i(\cdot , v_{-i})=c$ (see, e.g., \cite{noamchapter}). The combination of the two pillars justifies the focus on the social choice function: given $f$, one can easily determine whether it is implementable, and if so, the prices are (almost) unique.

Fadel and Segal \cite{fs} were the first to make the important observation that this paradigm breaks when computational considerations are taken into account. In other words, if computing the alternative chosen by $f$ is computationally ``easy'', can it be that determining how much each player has to pay is much harder? Note that the characterization theorems discussed above guarantee the existence of ``good'' payment functions, but they do not guarantee an efficient way to actually compute the prices.

Specifically, Fadel and Segal consider an implementable social choice function $f$ with communication complexity $cc(f)$. 
%Fix the valuations of the players $v_1,\ldots, v_n$. 
Denote by $cc_{IC}(f)$ the communication complexity of the implementation of $f$. The implementation of $f$ must output both the chosen alternative and the payments, so clearly $cc_{IC}(f)\geq cc(f)$. But can it be that computing the prices $P_1(\cdot), \ldots, P_n(\cdot)$ makes the computational task much harder, that is $cc_{IC}(f)>>cc(f)$?

Fadel and Segal showed that the gap is at most exponential: $cc_{IC}(f)\leq 2^{cc(f)}-1$. They also showed that the inequality $cc_{IC}(f)\geq cc(f)$ is strict by providing a specific $f$ for which $cc_{IC}(f)=cc(f)+1$. Although they were able to show that for bayesian implementations the gap indeed might be exponential, determining whether it is exponential for the basic setting of ex-post implementations was left as their main open question.

Babaioff, Blumrosen, Naor, and Schapira \cite{schapira} managed to narrow the gap and prove that for every $n$, there exists a function $f$ for $n$ players for which $cc_{IC}(f)\geq n\cdot cc(f)$.\footnote{In fact they write: ``We stress that achieving a better lower bound than the linear lower bound shown in this paper may be hard. The communication cost is known to be at most linear (in the number of players) for welfare-maximization objectives and in single-parameter domains (in FS)''. However, their interpretation of the results of Fadel and Segal is not accurate, since as mentioned by Fadel and Segal, their results assume that the type space is sufficiently small. As will be discussed later, we will be able to improve over this lower bound both for welfare maximization and for single parameter domains.} They also provided several single-parameter domains for which the gap is small: for every $f$ in these domains, $cc_{IC}(f)=O(cc(f))$.

\subsection*{Our Results I: Impossibilities}

Our first main result (Section \ref{lowerboundsec}) answers the open questions of \cite{schapira,fs} by showing that computing the payments might be significantly harder than computing the output:

\vspace{0.1in} \noindent \textbf{Theorem: } For every $k$, there exists a function $f$ for two players (or more, by adding players that do not affect the outcome) for which $cc(f)=O(k)$ and $cc_{IC}(f)=exp(k)$. Therefore, $cc_{IC}(f)=\exp(cc(f))$.\footnote{We stress that we show all truthful mechanisms for $f$ require at least $\exp(cc(f))$ bits, whereas the linear lower bound of \cite{schapira} applies only to the the normalized mechanism of $f$.}

\vspace{0.1in} \noindent In fact, the function $f$ that we provide is simple enough in the sense that it is single parameter. We note that a similar result was obtained concurrently and independently by \cite{RSTWZ20} (but their function is not single parameter). Roughly speaking, we construct a function $f:V_A
\times V_B\rightarrow \mathcal{A}$, where $\mathcal{A}=\{a_0,\ldots, a_{2^k}\}$. The domain of valuations is single parameter, and Alice's private information is $r_A\in [0,2^{k+1}-1]$. For each alternative $a_i\in \mathcal A$, let $w_A(a_i)=|\mathcal{A}|^{4ik}-1$. Alice's value for alternative $a_i$ is $r_A\cdot w_A(a_i)$. Bob is also a single parameter player but his valuation takes a simpler form: he is indifferent to the alternative chosen and his private information $r_B$ is also his value of each alternative. However, the number of possible values that $r_B$ can take is \emph{doubly} exponential in $k$. The function $f$ itself is defined by some arbitrary map that takes the possible values $r_B$ and projects each one to a different partitioning of the possible values of Alice to the $|\mathcal A|$ alternatives, making sure that each such partitioning is monotone: if $r,r'$ are two values such that $r\geq r'$, then $r$ is not mapped to a lower alternative than $r'$. The function $f$ takes the value of Alice and outputs the alternative that it belongs to according to the monotone map that is determined by Bob's value.

Computing $f$ is easy: Alice can send her private information $r_A$ ($k+1$ bits) and Bob can then compute the output of $f$ and announce it ($k+1$ bits). How about computing the payments? Bob is always indifferent to the chosen alternative, so his payment is always $0$. Computing the payment of Alice is a bit more subtle. Recall that by Myerson's formula the payment of Alice for an alternative $a$ is given by $P_A(v_A,v_B)=r_A\cdot w_i(f(v_A,v_B))-\int_0^{r_A}w_i(f(z\cdot w_i,v_B))dz$, where $v_B$ is Bob's valuation. Thus, the problem of computing the payments reduces to computing the integral in the formula. The crux of the proof is showing that even if we know that the outcome is $a_{2^k}$, the value of the integral is different for each map (this is why Alice's value for an alternative is obtained by multiplying her private information $r_A$ by a large number). Since each $r_B$ of Bob defines a different map and hence a different payment, the number of distinct prices for the alternative $a_{2^k}$ is doubly exponential. Standard arguments imply that at least $2^k$ bits are required to specify the payments, which completes the proof.

We note that $f$ is not a very natural, but we can build on it to show that sometimes even welfare maximization can be hard: there exists a multi-unit auction such that computing the welfare maximizing solution requires $k$ bits, but computing the payments requires $exp(k)$ bits. This result has one additional interesting feature: it provides an example of an auction domain where the approximation ratio to the social welfare achievable by non-truthful algorithms that use polynomial communication (in our case, the approximation ratio is $1$) is strictly better than the approximation ratio that can be achieved by truthful mechanisms that use only polynomial communication (in our case, we show that exponential communication is needed for a truthful mechanism to achieve an approximation ratio of $1$). This is only the second such example, following \cite{assadi2020separating} (other separations exist but in non-auction domains or in auctions with restrictions). Unlike all previous separations, in which the hardness is based on the hardness of computing the allocation, here computing the allocation is easy so the hardness stems from the additional overhead of computing the prices.

Quite remarkably, the function $f$ demonstrates that even if computing the output requires only $k$ bits, the number of possible payments in the truthful implementation might be as large as $exp(exp(k))$. In fact, one can see that the possible number of distinct prices was the decisive factor in determining the communication complexity of a mechanism for $f$. This is no coincidence. Denote by $P_f$ the maximum possible payments for a single alternative that any player might face. Then, the communication complexity of truthfully implementing an implementable function $f$ \emph{for two players} can be determined up to a constant multiplicative factor:
$$
\frac {cc(f)+\log P_f} 2 \leq cc_{IC}(f) \leq cc(f)+2\log P_f 
$$
The left inequality holds since obviously $cc_{IC}(f)\geq cc(f)$ and since $cc_{IC}(f)\geq \log P_f$, because $\log P_f$ bits are needed to specify which price the player has to pay out of the possible $P_f$ prices. The right inequality holds since we can use $cc(f)$ bits to compute the output of $f$, and then each of the two players uses (at most) $\log P_f$ bits to specify the price of the other player (recall that by the taxation principle, the price of an alternative for a player depends only on the valuations of the other players).

We thus have that for any two-player function $cc_{IC}(f)=poly(cc(f),\log P_f)$. Note that this characterization is tight in the sense that it is easy to come up with examples where $cc_{IC}(f)>> \log P_f$, and as discussed above $\log P_f$ is also necessary for characterizing the communication complexity. 

When there are three players or more, this (or similar) characterization no longer holds. In fact we can show a function $f$ with a payment scheme $P$ such that $\log|\Ima P|= poly(cc(f))$, but computing $P$ (or every other payment scheme that implements $f$) is significantly harder than computing $f$ alone.

\vspace{0.1in} \noindent \textbf{Theorem: } For every $k$, there exists a function $f$ for three players (or more, by adding players that do not affect the outcome) and a payment that implements it $P$  for which $cc(f)=O(k)$, $\log|\Ima P|=poly(k)$, and $cc_{IC}(f)=exp(k)$. Therefore, $cc_{IC}(f)=\exp(cc(f))$.

\vspace{0.1in} \noindent The proof is very different than the previous proof. Rather than basing the hardness on the number of payments, the hardness stems from carefully constructing the function so that determining the prices for Alice requires to decide whether the bit representations of the types of Bob and Charlie share a common $1$ bit, whereas computing $f$ requires to decide whether a \emph{single} specific bit intersects.  Of course, determining whether Bob and Charlie share a common $1$ bit is just the disjointness function in disguise, which implies that computing the payments is indeed hard.

\subsection*{Our Results II: Algorithms for Payment Computation}
%\colorbox{yellow}{TODO - correct section}
%In Section 
%\ref{algorithmssection}, 
We proceed with developing algorithms for payment computation. Our algorithms come in three different flavours: truthful in expectation implementations of deterministic social choice functions, deterministic implementations of single parameter functions, and deterministic implementations of multi-parameter functions that satisfy uniqueness of payments.

We observe that if $f$ is  an implementable social choice function, then although $cc_{IC}(f)$ might be exponential in $cc(f)$, if we compromise on truthful-in-expectation implementation,\footnote{Recall that, roughly speaking, in a truthful in expectation mechanism each player has a strategy that maximizes his \emph{expected} profit regardless of the strategies of the other players, where the expectation is taken over the random coins of the mechanism.} the payment might be computed very efficiently. In fact, if we let $cc_{TIE}(f)$ be the communication complexity of implementing $f$ as a truthful in expectation mechanism, we prove that $cc_{TIE}(f)=poly(n,cc(f))$ for single parameter domains and for multi-parameter domains that are convex. For single parameter domains, we rely on the (well known) observation that the expected value of the integral in Myerson's formula can be estimated by measuring the height of the integral at a random point. 
For convex multi-parameter domains, we develop another algorithm reyling on a characterization of the payments in scalable domains by Babaioff, Kleinberg, and Slivkins \cite{babaioff2015truthful}, by showing that convex domains are essentially scalable. 
%which we essentially extend to hold also for convex domains. 

In the rest of the paper, we return to consider the fundemental question of payment computation in (deterministic) ex-post equilibrium. Babaioff et al. \cite{schapira} consider several simple single parameter problems. All of their problems are \emph{binary}: for each player $i$, the set of alternatives is divided into a set of ``winning'' alternatives for which his value is his private information $r_i$, and a set of ``losing'' alternatives for which his value is $0$. Babaioff et al. \cite{schapira} provide algorithms for payment computation for some specific settings. Our first algorithm provides a general polynomial upper bound for all binary problems:

\vspace{0.1in} \noindent \textbf{Theorem: } Let $f$ be an implementable social choice function for $n$ players in a binary single parameter domain. Then, $cc_{IC}(f)\leq O(n\cdot cc^2(f))$.

\vspace{0.1in} \noindent In fact, our algorithm extends to a much more general single parameter setting that may have many alternatives, like the setting of our impossibility result discussed above. In this case we show that $cc_{IC}(f)\leq O(n\cdot cc^2(f)\cdot |\mathcal{A}|)$, where $\mathcal{A}$ is the set of alternatives. This bound is tight in the sense that if we omit $cc(f)$ or $|\mathcal{A}|$ from the RHS, $cc_{IC}(f)$ might become much bigger than the RHS. 

We then proceed to considering multi-parameter settings. These turn out to be quite challenging. However, we do provide an algorithm for those domains as well, assuming ``uniqueness of payments'', i.e., that the payment functions are uniquely determined by the allocation function (up to a constant). Most interesting domains (combinatorial auctions, scheduling, etc.) satisfy uniqueness of payments. 

\vspace{0.1in} \noindent \textbf{Theorem: } Let $f$ be an implementable social choice function for $n$ players that satisfies uniqueness of payments. Then, $cc_{IC}(f)\leq poly(n, cc(f), |\mathcal{A}|)$.

\vspace{0.1in} \noindent To prove this theorem, we first prove that there exists a \emph{non-deterministic} algorithm that computes the payment of player $i$. We then leverage this result and the fact that non-deterministic and deterministic communication complexity are polynomially related to establish our upper bound.
We show that for every player $i$ and every price, the prover can send  $\mathcal{O}(|\mathcal{A}|^2)$ types in $V_i$ that serve as a non-deterministic witness. The proof of the theorem consists in explicitly describing those types, showing that they suffice and that they can be described succinctly. 
\subsection*{Our Results III: The Hardness of Computing the Payments in a Menu}

We now change gears and consider a slightly different but very related problem (Section \ref{sec-hardness}). Up until now we assumed that we are given an instance $(v_1,\ldots, v_n)$ and we want to compute the payment of each player. However, the taxation principle asserts that each truthful mechanism can be seen as follows: each player $i$ is facing a menu that specifies a price for each alternative. The output of a truthful mechanism is an alternative that maximizes the profit, i.e., maximizes $value(a) - price(a)$ for each player. The taxation principle leads to a definition of taxation complexity,\footnote{The taxation complexity of a mechanism is $\log (\max_i |M_i|)$, where $M_i$ is the set of possible menus player $i$ might face.} which was shown to characterize the communication complexity of truthful mechanisms in many settings \cite{dobzinski2016computational}. The notion of taxation complexity was crucial in establishing a lower bound on the communication complexity of truthful approximation mechanisms in the recent breakthrough of Assadi et al. \cite{assadi2020separating}.

Consider the notion of a ``constructive taxation principle'' or ``menu reconstruction'' \cite{dobzinski2016computational}: an algorithm that efficiently finds the menu that $v_{-i}$ presents to player $i$. The basic building block of this algorithm is a subroutine $price(\cdot)$ that assumes that the input of each player $i'\neq i$ is $v_{i'}$, gets an alternative $a$ and returns the price of $a$ in the menu induced by the truthful mechanism $M$. We have efficient and constructive taxation principle whenever $cc(price)=poly(cc(M))$.

We pinpoint the hardness of $price(\cdot)$ on  deciding whether an alternative $a$ is \emph{reachable}, i.e., whether there exists $v$ such that $f(v,v_{-i})=a$. We denote this function with $reach(\cdot)$.
We show that if $reach(\cdot)$ is \textquote{easy}, $price(\cdot)$ is also \textquote{easy}, i.e.:  $cc(price)\le poly(cc(reach), cc(M), n)$. We use this observation to show that for all the mechanisms of player decisive functions, $cc(price)\le poly(n,cc(M))$.  
Furthermore, we show an instance $M=(f,P)$ where $cc(reach)=exp(cc(M))$ and prove that this gap is tight.

\section{Preliminaries}   	 \label{setupmyersonsection} 
\paragraph{Truthfulness}
We consider settings with $n$ players. Each player $i$ has a valuation function $v_{i}:\mathcal{A}\to \mathbb{R}$ which is his private information. Let $ V_{i}$ be the set of all possible valuations of player $i$. A mechanism $M$ consists of a social choice function $f\colon V_{1}\times \cdots \times V_{n} \to \mathcal{A}$, where $\mathcal A$ is the set of possible alternatives, and a payment function $P_{i}\colon V_1 \times\cdots\times V_n \to \mathbb{R}$ for each player $i$. A mechanism is \emph{ex-post incentive compatible} (or truthful) if for each player $i$, every valuations profile of the other players $v_{-i}\in V_{-i}$ and every $v_i$, $v'_i\in V_i$, it holds that:
$$
v_i(f(v_i,v_{-i}))-P_i(f(v_i,v_{-i})) \geq v_i(f(v'_i,v_{-i}))-P_i(f(v'_i,v_{-i}))
$$
$f$ is called \emph{implementable} if for some $P_1,\ldots, P_n$ the resulting mechanism is ex-post incentive compatible. We denote the image of a payment function $P_i$ with $\Ima P_i$.

In this paper we give special treatment to single parameter domains. A domain of a player $V_{i}$ is \emph{single parameter} if there exists a public function $w_{i}:\mathcal{A}\to \mathbb{R}$ and a set of real numbers $R_i\subseteq\mathbb{R}$ such that $V_i=\{r\cdot w_i(\cdot)\hspace{0.25em}|\hspace{0.25em}r\in R_i \}$.  A single parameter domain $V_i$ is \emph{binary} if its public function $w_i$ satisfies that $\Ima w_i=\{0,1\}$.   
 If $V_1, \ldots, V_n$ are all single parameter domains, we say that $f\colon V_{1}\times\cdots \times V_{n} \to \mathcal{A}$ is single parameter. In particular, since we can assume that the private information of player $i$ is $r_{v_i}$, we often identify $v_i$ with $r_{v_i}$ and slightly abuse notation by writing, e.g., $v_i>v'_i$ where we mean $r_{v_i}>r_{v'_i}$.

%The definition allows negative elements in the image of the public functions $w_{i}$ and in the type space $V_{i}$. Also, a zero type does not necessarily exist in $V_{i}$.
A social choice function $f$ is \emph{monotone} with respect to player $i$, if $V_{i}$ is a single parameter domain and for every $v_{-i}\in V_{-i}$:
		\begin{equation*} 
		r_{v_{i}'} > r_{v_{i}} \implies w_{i}(f(v_{i}',v_{-i})) \ge w_{i}(f(v_{i},v_{-i})) 
		\end{equation*}
		$f$ is monotone if it is monotone with respect to each of its players.
		
Let $M=(f,P)$ be a mechanism over a domain, where each $V_{i}$ is single parameter and $0\in R_i$. $M$ is \emph{normalized} if for each player $i$ and every $v_{-i} \in V_{-i}$, $r_{v_i}=0 \implies P_{i}(v_{i},v_{-i})=0$. 

The following proposition is well known \cite{myerson1981optimal}:

\begin{proposition}[Monotonicity and Myerson's Payment Formula]\label{myersonlemmainterval}
	Let $V=V_1\times \cdots \times V_n$ be a single parameter domain. Then, a social choice function $f\colon V\to \mathcal{A}$ is implementable if and only if it is monotone. If $R_i=[0,b_i]$,\footnote{$b_i$ might be equal to $\infty$.} then the \emph{unique} payment rule of player $i$ that 
%	implements $f$ and
	 satisfies normalization is given by:
	\begin{equation} \label{myersonintegral}
	P_{i}(v_{i},v_{-i})=r_{v_{i}}\cdot w_{i}(f(v_i,v_{-i})) - \int_{0}^{r_{v_{i}}}w_{i}(f(z\cdot w_i,v_{-i}))dz
	\end{equation} 
		%\item Let $f$ be a monotone social choice function, and let player $i$ satisfy that  $V_{i}=[0,b_i]$ where $b_i\in \mathbb{R_{+}}$. Then, there exists a unique payment scheme $P_{i}$ that implements $f$ for player $i$. It is given by (\ref{myersonintegral}).\footnote{This corollary follows from the proof of Myerson lemma, since it starts by fixing a type $v_{-i}$ without assuming anything about it.}     
%	\end{enumerate}
\end{proposition}

\paragraph{Communication Complexity}
In this paper communication complexity refers to the number-in-hand model where $v_i$ is the input of player $i$. We denote by $cc(f)$ the communication complexity of a function $f$. We denote by $cc_{IC}(f)$ the cost of the most efficient mechanism that implements $f$, including payments.\footnote{The definition of \cite{schapira} is slightly different: there, $cc_{IC}(f)$ is the cost of the most efficient \emph{normalized} mechanism for $f$.}
%	\begin{equation*} 
%	cc_{IC}(f)=\min_\text{$M$ implements $f$}cc(M)
%	\end{equation*} 
%	 Fadel and Segal \cite{fs} have proven:
%	\begin{theorem} \label{fstheorem} For every social choice function $f$, $cc_{IC}(f)\le \exp(cc(f))$.
%	\end{theorem}
\begin{remark}\label{finiteccfremark}
Our focus in this paper is understanding that gap between the communication complexity of computing the payments and the communication complexity of computing $f$. Therefore, for all the social choice functions considered in the paper we assume that $cc(f)$ is finite (otherwise, there is no finite communication protocol for $f$ and understanding the gap makes little sense).
\end{remark}

\section{The Cost of Payment Computation is Exponential}\label{lowerboundsec}

Recall that Fadel and Segal \cite{fs} showed that for every social choice function $f$ we have that $cc_{IC}(f)\leq exp(cc(f))$. In this section we solve their main open question and show that their bound is tight, that is, there exists a social choice function $f$ such that $cc_{IC}(f)= exp(cc(f))$. In fact, we provide two proofs by constructing two social choice functions, each function highlights a different source of hardness of payment computation.

The first source of hardness is the fact that the number of prices that a player might see in an implementation of a social choice function $f$ with $cc(f_k)=k$ is doubly exponential in $k$. Therefore, just specifying the payments requires $exp(k)$ bits, which immediately implies that $cc_{IC}(f_k)=exp(k)$. 

If there are only two players, we show that this is the only source of hardness in the sense that payment computation becomes easy when the number of payments is not huge. However, when there are more than two players we show that even when the number of payments is small, payment computation might be hard because of the interaction between the players.

One possible criticism about those results is that the functions that we construct are quite contrived. Thus, we conclude by showing a welfare maximizer in a multi-unit auction that satisfies that $cc_{IC}(f)=\exp(cc(f))$. 

\subsection{Proof I: Hardness via the Number of Payments} \label{hardnessnumberofpaymentssec}

In the two player case, we are able to fully characterize the relationship between $cc_{IC}(f)$ and $cc(f)$. For every implementable social choice function $f$, let $P\colon V_1\times \cdots\times V_n \to \mathbb{R}^{n}$ be the payment scheme of the most efficient mechanism for $f$, i.e.  the one that satisfies $cc(f,P)=cc_{IC}(f)$. Let $P_{f}$ be the maximum number of prices for an alternative when using $P$.  Formally:
	\begin{equation*} \label{pfdefinition}
	P_{f} = \max_{i\in[N]} \hspace{0.5em} \max_{a\in \mathcal{A}} \hspace{0.5em} \big|\{p \hspace{0.5em} | \hspace{0.5em} \exists v\in V \hspace{0.5em} \text{s.t.} \hspace{0.5em} f(v)=a, \hspace{0.30em} P_{i}(v)=p  \}\big|
	\end{equation*}
	
	\begin{proposition} \label{twoplayertheorem}
		Let $f$ be an implementable social choice function for two players.
		Then:
		\begin{equation}\label{characterizationeq}
		\dfrac{cc(f)+ \log (P_{f})}{2} \le cc_{IC}(f)\le cc(f) + 2\log (P_{f})
		\end{equation}
	\end{proposition}
%We stress Theorem \ref{twoplayertheorem} holds for all settings, multi-parameter as well as single-parameter. 
	\begin{proof}
		Obviously, $cc_{IC}(f)\ge cc(f)$. 
		 Also, $cc_{IC}(f)\ge \log(P_{f}) $, because
		$cc_{IC}(f)\ge \log |\Ima (f,P)|\ge \log (P_{f})$. For the RHS, denote the payment functions of players $1$ and $2$ with $P_1$ and with $P_2$, respectively. By the taxation principle, $P_1$ can be reformulated as a function of the alternative chosen and of $V_2$, and analogously $P_2$ is a function of the alternative chosen and of $V_1$.
		
We use this reformulation to explicitly provide a protocol for a truthful implementation fo $f$. Let $\pi^f$ be the most efficient communication protocol of $f$. Fix types $v_1\in V_1, v_2\in V_2$. The players first run $\pi^{f}(v_1,v_2)$, so they both know $f(v_{1},v_{2})=a$. By the above, once the alternative is known, player $1$ knows $P_2(a,v_{1})$ and player $2$ knows  $P_1(a,v_{2})$. The players can now send to each other those payments, using at most $2\log(P_{f})$ bits. 
	\end{proof}

As a direct corollary, to prove an exponential gap between $cc(f)$ and $cc_{IC}(f)$ when there are only two players, we must construct an $f$ in which the number of possible payments $P_f$ is \emph{doubly} exponential in $cc(f)$. We now construct such an $f$, which gives us the first proof of our main result. We note that the $f$ that we construct is quite contrived. However, in Section \ref{vcgsection} we use the same $f$ to prove that payment computation is hard even if we want to maximize the welfare in a multi-unit auction setting.

Fix some integer $k$. In our setting there are two players, Alice and Bob, and $2^{k}+1$ alternatives: $\mathcal{A}=\{a_{0},a_{1},\ldots,a_{2^{k}}\}$. The domains of Alice and of Bob are single parameter: $r_A\in [0,2^{k+1}-1]$ and $r_B\in \{0,1,\ldots,l-1\}$, where $l=
%\binom{2^k+|\mathcal{A}|-2}{|\mathcal{A}|-1}=
\binom{2^{k+1}-1}{2^k}$. 
Notice that the domain of Bob's valuations, $V_B$, is  of size $l$. 
	  The value of Alice for alternative $a_i$ is $r_A\cdot (|\mathcal{A}|^{4ik}-1)$, i.e., $w_A(a_i)=|\mathcal{A}|^{4ik}-1$. Bob's value for all alternatives is identical and equal to his private information $r_B$ ($w_B\equiv 1$). 
	  Let $\mathcal{C}$ be the following set of $(2^{k}+1)$-dimensional vectors:
	 \begin{equation} \label{cdefinition}
	 \mathcal{C}=\Big\{ \big(c_{0},\ldots,c_{|\mathcal{A}|-1}\big) \Big| \forall i \hspace{0.3em}  c_i\ge 1 \hspace{0.3em} \text{and} \hspace{0.3em} c_i\in \mathbb{N}, \hspace{0.3em} \sum_{i=0}^{|\mathcal{A}|-1}c_{i} = 2^{k+1} 
%	 \hspace{0.3em}, c_{|\mathcal{A}|-1}\ge 1,\hspace{0.3em} c_{i}\in \mathbb{N} 
\Big\}
	 \end{equation}
	 Each vector $c\in \mathcal{C}$ defines a function $c:\{0,1,\ldots,2^{k+1}-1\}\to\mathcal{A}$, where $c_{j}$ is the number of integers which are assigned to alternative $a_{j}$: $c=\big(c_{0},\ldots,c_{|\mathcal{A}|-1}\big)$ maps the integers in $\{0,1,\ldots,c_{0}-1\}$ to $a_{0}$, the integers in $\{c_{0},c_{0}+1,\ldots,c_{0}+c_{1}-1\}$ to $a_{1}$ and so on. Each $c\in \mathcal{C}$ is monotonically increasing in the sense that it maps larger integers to alternatives with no smaller index. 
	 Note that:
	 \begin{equation} \label{csize}
	 |\mathcal{C}|=\binom{2^{k+1}-|\mathcal{A}|+|\mathcal{A}|-1}{|\mathcal{A}|-1}=
	 \binom{2^{k+1}-1}{2^k}=|V_{B}|
	 \end{equation}
	 Since this is the number of ways to match the the integers in $\{0,1,\ldots,2^{k+1}-1\}$ to alternatives in a monotone way, with the constraint that each alternative is matched with at least one integer (for all $i$, $c_{i}\ge 1$). It follows that there exists a bijective function between $V_{B}$ and the set $\mathcal{C}$. Let $map\colon V_{B}\to \mathcal{C}$  be such bijection. We define $f_k(v_{A},v_{B})=map(v_{B})(\lfloor r_{v_{A}} \rfloor)$. In words, computing $f_k(v_{A},v_{B})$ is done by first computing $map(v_{B})$ which returns a function-vector $c\in \mathcal{C}$. Afterwards, we apply $c$ to the integer $\lfloor r_{v_{A}} \rfloor$, which returns an alternative. 	
\begin{theorem} \label{mostimportanttheorem}
	For the $f_k$ defined above, $cc(f_k)=\mathcal{O}(k)$,  whereas $cc_{IC}(f_k)\ge\exp(k)$.
\end{theorem}	 
\begin{proof}
Observe that $cc(f_k)=\mathcal{O}(k)$ since $f_k$ can be computed by a simple protocol where Alice sends to Bob $\lfloor r_{v_{A}}\rfloor$ using $k+1$ bits, and then Bob computes $f_k$ and sends the outcome to her, using $\log |\mathcal{A}|$ bits. We now show that $f_k$ can be truthfully implemented, then we will analyze $cc_{IC}(f_k)$.
	\begin{lemma} \label{monotonelemma}
		$f_k$ is monotone and hence implementable. 
	\end{lemma}
	\begin{proof}
		$w_{B}(\cdot)$ is constant so $f_k$ is obviously monotone with respect to Bob. In order to show monotonicity with respect to Alice as well, we fix $v_{B}\in V_{B}$ and two types $r,r'\in [0,2^{k+1}-1]$ such that $r>r'$. Denote the valuations $r\cdot w_A(\cdot)$ and $r'\cdot w_A(\cdot)$ with $v$ and with $v'$, respectively. Denote $map(v_{B})$ with $c$, and define $index:\mathcal{A}\to\mathbb{N}$ as  $index(a_{i})=i$. We wish to prove $w_A(f_k(v,v_B))\ge w_A(f_k(v',v_B))$.

$r>r'$ clearly implies $\lfloor r\rfloor \ge \lfloor r'\rfloor$. By definition, $c=map(v_B)$ is monotonically increasing with respect to the index of alternative, so $index (c \lfloor r\rfloor) \ge index (c\lfloor r'\rfloor)$. $w_A$ assigns greater values to alternatives with higher index, so  $w_A(c \lfloor r\rfloor) \ge w_A(c \lfloor r'\rfloor)$. By the definition of $f_k$, we get that $w_A(f_k(v,v_B))\ge w_A(f_k(v',v_B))$.       
	\end{proof}
We now analyze the hardness of computing the payments of Alice. By Proposition \ref{myersonlemmainterval}, her \emph{unique} normalized payment scheme is:  
	\begin{equation} \label{alicepayments} 
	P_{A}(v_{A},v_{B})=r_{v_A}\cdot w_{A}(f_k(v_{A},v_{B})) - \int_{0}^{r_{v_A}}w_{A}(f_k(z\cdot w_A,v_{B}))dz
	\end{equation}
\begin{claim} \label{ccicbiggerccaclaim}
	$ cc(P_{A})\le 2\cdot cc_{IC}(f_k)$.
\end{claim}
\begin{proof}
	Let $M^{\ast}=(f_k,P^{\ast})$ be the most efficient mechanism for $f_k$, i.e. $cc(M^{\ast})=cc_{IC}(f_k)$. 
	Denote Alice's valuation when $r_A=0$ with $v_0$, i.e., $v_0\equiv 0\cdot w_A$. 
Run $M^\ast$ on the instances $(v_A,v_B)$ and $(v_{0},v_B)$ to obtain $P_A^\ast(v_A,v_B)$ and $P_A^\ast(v_{0},v_B)$.
By Proposition \ref{myersonlemmainterval}, Alice has a unique normalized payment scheme, so we get that  $P_A(v_A,v_B)=P_A^\ast(v_A,v_B)- P_A^\ast(v_{0},v_B)$.   
\end{proof}	 
We now move to the main part of the proof which is   
showing that the image of $P_{A}$ is \textquote{large}. We start with the following lemma. Recall that $\langle  b, c\rangle$ stands for the dot product of the vectors $ b$ and $ c$. 
	\begin{lemma} \label{mainlemma}
		Let $w$ be the vector $\big(w_{A}(a_{0}),w_{A}(a_{1}),\ldots,w_{A}(a_{2^{k}})\big)$.
		For every two vectors $c\neq c' \in \mathcal{C}$, $\langle w,c\rangle \neq \langle w,c'\rangle$.
	\end{lemma}
	\begin{proof}
		Denote with $j$ the largest index where $c$ and $c'$ differ. Assume without loss of generality that $c_{j}>c_{j}'$. If $j=0$, it means that $c$ and $c'$ differ in only one coordinate, so their dot products with $w$ cannot be equal to each other. Hence, we can assume from now on that $j\ge 1$. We will show that $\langle w,c\rangle > \langle w,c'\rangle$:
		\begin{align}
		\langle w,c\rangle - \langle w,c'\rangle &= \sum_{i=0}^{|\mathcal{A}|-1}w_{i}\cdot (c_{i}-c'_{i}) \nonumber \\ 
		&= \sum_{i=0, c_{i}>c_{i}'}^{|\mathcal{A}|-1}w_{i}\cdot (c_{i}-c'_{i}) + \sum_{i=0, c_{i}'>c_{i}}^{|\mathcal{A}|-1}w_{i}\cdot (c_{i}-c'_{i}) \nonumber \\ &\ge  
		w_{j}\cdot \underbrace{(c_{j}-c_{j}')}_{\substack{{\ge 1}\\\text{since $c_{j},c'_{j}\in\mathbb{N}$}\\\text{and $c_{j}>c_j'$}}}
		+ \sum_{i=0, c_{i}'>c_{i}}^{|\mathcal{A}|-1}w_{i}\cdot (c_{i}-c'_{i}) \label{blacktriangle1} \\
		&\ge \underbrace{w_{j}}_{=|\mathcal{A}|^{4jk}-1} + \sum_{i=0, c_{i}'>c_{i}}^{|\mathcal{A}|-1}w_{i}\cdot (c_{i}-c'_{i}) \nonumber\\
		&>|\mathcal{A}|^{4jk}-1 -|\mathcal{A}|^{4jk}+1 = 0 \label{blacktriangle2}
		\end{align}
		which completes the proof. 
		(\ref{blacktriangle1}) holds because there are only positive summands in $\sum_{i=0, c_{i}>c_{i}'}^{|\mathcal{A}|}w_{i}\cdot (c_{i}-c'_{i})$, one of them is $w_{j}\cdot (c_{j}-c'_{j})$. 
%		Hence, the summation is at least as large as any of its components. 
		We now  explain (\ref{blacktriangle2}),  by proving that  $-\big(\sum_{i=0, c_{i}'>c_{i}}^{|\mathcal{A}|-1}w_{i}\cdot (c_{i}-c'_{i}) \big) < |\mathcal{A}|^{4jk}-1$. 
			\begin{multline*}
			-(\sum_{i=0, c_{i}'>c_{i}}^{|\mathcal{A}|-1}w_{i}\cdot (c_{i}-c'_{i})) = \sum_{i=0, c_{i}'>c_{i}}^{|\mathcal{A}|-1} \underbrace{w_{i}}_{\substack{\text{$\le w_{j-1}$, since $j$ is}\\\text{the maximal coordinate}\\ \text{where $c,c'$ differ and $c_{j}>c_{j}'$}}} \cdot \underbrace{(c_{i}'-c_{i})}_{\le 2^{k+1}} 
			\underbrace{\le}_{\substack{\text{$j\ge 1$, so}\\ \text{$w_{j-1}$ is well defined}}}  
			 |\mathcal{A}| \cdot w_{j-1} \cdot 2^{k+1} \\
			<|\mathcal{A}| \cdot (|\mathcal{A}|^{4(j-1)k}-1)\cdot\underbrace{ 2^{2k}}_{<|\mathcal{A}|^2}
			<  |\mathcal{A}|^{4jk-4k+3}-|\mathcal{A}|^3 
			\underbrace{<}_{k\ge 1} |\mathcal{A}|^{4jk}-1 
%			\underbrace{<}_{k\ge 1}
		\end{multline*}
	\end{proof}
	\begin{claim}\label{claimccic}
		Let $v_A\in V_A$ be a valuation such that $r_{v_{A}}=2^{k+1}-1$. Then, for all $v_1,v_2\in V_B$:
			\begin{equation} \label{threeandhalf}
		v_{1} \neq v_{2} \implies P_{A}(v_{A},v_{1}) \neq P_{A}(v_{A},v_{2})
		\end{equation} 
		As a corollary, if we reformulate $P_A$ as a function of the alternative and of Bob's value, we get that:
		\begin{equation*}
		v_{1} \neq v_{2} \implies P_{A}(a_{2^k},v_{1}) \neq P_{A}(a_{2^k},v_{2})
		\end{equation*} 
		\end{claim}
	We will use the corollary in Subsection \ref{vcgsection}.
%	\begin{claim}\label{claimccic1}
%		$|\Ima P_{A}|\ge \binom{2^{k+1}-1}{2^k}$. 
%	\end{claim}
	\begin{proof}
		Note that all elements in $\mathcal{C}$ are monotone and satisfy that $c_{2^{k}}\ge 1$, thus  for every $v_{B}\in V_{B}$, if $r_{v_{A}}=2^{k+1}-1$, then $f_k(v_A,v_{B})=a_{2^{k}}$. In words, Alice always gets alternative $a_{2^k}$ when bidding her highest value. Combining (\ref{threeandhalf}) with the payment formula in (\ref{alicepayments}), we get the following logical equivalences:		
		\begin{eqnarray*}
		(2^{k+1}-1)\cdot w_{A}(a_{2^{k}}) - \int_{0}^{2^{k+1}-1}w_{A}(f_k(z\cdot w_A,v_{1}))dz&\overset{\text{?}}{=}&
		(2^{k+1}-1)\cdot w_{A}(a_{2^{k}}) - \int_{0}^{2^{k+1}-1}w_{A}(f_k(z\cdot w_A,v_{2}))dz\\	\int_{0}^{2^{k+1}-1}w_{A}(f_k(z,v_{1}))dz& \overset{\text{?}}{=}& \int_{0}^{2^{k+1}-1}w_{A}(f_k(z,v_{2}))dz \\
		\int_{0}^{2^{k+1}-1}w_{A}\big(map(v_{1})\lfloor z \rfloor\big) dz &\overset{\text{?}}{=}& \int_{0}^{2^{k+1}-1}w_{A}\big(map(v_{2})\lfloor z \rfloor\big) dz
		\end{eqnarray*}
		Denote $map(v_{1})$ with $c_{1}$ and $map(v_{2})$ with $c_{2}$:
		\begin{eqnarray*}
		\int_{0}^{2^{k+1}-1}w_{A}\big(c_{1}\lfloor z \rfloor \big) dz &\overset{\text{?}}{=}& \int_{0}^{2^{k+1}-1}w_{A}\big(c_{2}\lfloor z \rfloor \big) dz\\
		\langle c_{1}, w\rangle &\overset{\text{?}}{=}& \langle c_{2}, w\rangle 
		\end{eqnarray*}
		The last transition holds since the integral of $w_{A}\big(c_{1}(\cdot )\big)$ over the interval $[0,2^{k+1}-1]$ equals to $\langle w,c_{1} \rangle$ where $w=\big(w_{A}(a_{0}),w_{A}(a_{1}),\ldots,w_{A}(a_{2^{k}})\big)$, and the same clearly applies also to the RHS. Note that $c_{1}$ and $c_{2}$ are interpreted as functions in the uppermost equation, and as vectors in the lower equation. Recall that $map(\cdot)$ is one-to-one, so $v_{1}\neq v_{2}$ means that $c_{1} \neq c_{2}$. By Lemma \ref{mainlemma}, $c_{1}\neq c_{2}\implies\langle c_{1}, w\rangle \neq \langle c_{2}, w\rangle$. Therefore,  $P_{A}(v_A,v_{1})\neq P_{A}(v_A,v_{2})$, as required.
		
		As for the corollary, since bidding $v_A$ such that $r_{v_A}=2^{k+1}-1$  guarantees alternative $a_{2^k}$, it is immediate from the taxation principle that:
		\begin{equation*}
		P_A(a_{2^k},v_1)= P_A(v_A,v_1)\neq P_A(v_A,v_2)= P_A(a_{2^k},v_2) \implies
		P_A(a_{2^k},v_1)\neq 
		P_A(a_{2^k},v_2)
		\end{equation*}
		    
	\end{proof}
\begin{corollary} \label{ppaboundcorollary}
	$cc(P_{A})\ge 2^{k}$.
\end{corollary}
\begin{proof}
	By Claim \ref{claimccic}, different values of Bob induce a different payment for Alice whenever her she bids the valuation $v_A$ such that $r_{v_A}=2^{k+1}-1$. Thus, $|\Ima P_A|\ge |V_B|$, where 
		 $|V_{B}|= \binom{2^{k+1}-1}{2^k}$.
 Thus, we deduce that every protocol that computes $P_{A}$ has at least $\binom{2^{k+1}-1}{2^k}$ leaves, so $cc(P_{A})$ is bounded from below by $\log \binom{2^{k+1}-1}{2^k}$. Hence:
\begin{equation*}
cc(P_{A})\ge \log \binom{2^{k+1}-1}{2^k} \ge 2^k \cdot \log\frac{2^{k+1}-1}{2^k} 
%\underbrace{\ge}_{|\mathcal{A}|=2^{k}+1} 2^{k}\cdot \log\frac{2^{k+1}-1}{{2^{k}}}
\approx 2^{k}\cdot \log 2 = 2^{k}
\end{equation*}
where the second inequality is due to the binomial bound $\binom{n}{k} \ge (\frac{n}{k})^{k}$.
\end{proof}
To conclude the proof of Theorem \ref{mostimportanttheorem}, combining Claim \ref{ccicbiggerccaclaim} and Corollary \ref{ppaboundcorollary} yields that $cc_{IC}(f)\ge 2^{k-1}$. In contrast, as discussed above, $cc(f) \le k+1+\log |\mathcal{A}|\le 3k$.
\end{proof}

\subsection{Proof II: Hardness via Interaction}\label{polylogsec} 

We now show that if there are more than two players, then payment computation might be hard even if the number of possible payments is small. The idea is to construct a social choice function such that the payment is determined by the number of bit intersections of Bob's and Charlie's inputs (note that determining this number is harder than solving the disjointness problem). The challenge is to design such an $f$ with the additional property that $cc(f)$ is still small. We achieve that by constructing $f$ in which the chosen alternative depends only on $v_A$ and a constant number of bits of Bob and Charlie. That is, determining the output can be done ``locally'' but determining the payments is done ``globally''.

	\begin{theorem} \label{nopolylogthm}
		For every integer $k\ge 1$, there exists a single parameter social choice function $f_k$ over three players and $\mathcal{O}(k)$ alternatives, 
		where $cc(f_k)=\Theta(\log k)$, $cc_{IC}(f_k)=\Omega (k)$ and $\Ima P_{A}= \mathcal{O}(k)$, where $\Ima P_A$  is the image of the normalized payment function for Alice (by Proposition \ref{myersonlemmainterval}). 
	\end{theorem}
\begin{proof}
We describe the function $f_k$. For every integer $k\ge 1$, we define the set of alternatives as $\mathcal{A}=\{0,1,...,k\}$. There are three players, Alice, Bob and Charlie, with single parameter domains. Alice's private information is $r_A\in [0,k-1]$ and her value for each alternative $a$ is $r_A\cdot a$, i.e., $w_A(a)=a$ for every alternative. The private information of Bob and Charlie is $r_B,r_C\in \{0,1\}^k$.
Their values for all alternatives are identical and equal to the integer representations of their private information, so we use $v_B$ and $r_B$ interchangeably, and the same applies to $v_C$ and $r_C$. 
%In other words,  $w_{B}\equiv w_{C}\equiv 1$. 
 Denote with $v_B(j)$ and $v_C(j)$ the $j'$th bits of $v_B$ and of $v_C$. 
   We define $f_k\colon V_A\times V_B\times V_C \to \mathcal{A}$ as follows. For every $v_A \in V_A,v_B \in V_B,v_C \in V_C$:
		\begin{equation*}
		f_k(v_A,v_B,v_C)=
		\begin{cases}
		\lfloor r_{v_A} \rfloor +1   &\qquad 
		v_B(\lfloor r_{v_A}\rfloor) = v_C(\lfloor r_{v_A} \rfloor) = 1 \\
		\lfloor r_{v_A} \rfloor
		&\qquad \text{otherwise.}
		\end{cases} 
		\end{equation*} 
\begin{lemma}
	$f_k$ is monotone and hence implementable. 
\end{lemma}
\begin{proof}
	$f_k$ is clearly monotone with respect to Bob and Charlie. In order to prove monotonicity with respect to Alice, we fix $v_{A},v_{A}'\in V_A$, $v_B\in V_B$ and $v_c\in V_c$ such that $r_{v_{A}'}>r_{v_{A}}$. 
%	For brevity, we denote $r_{v_{A}}'$ with $r'$ and $r_{v_{A}}$ with $r$.
	  We want to show that $w_A(f_k(v_{A}',v_B,v_C))\ge w_A(f_k(v_{A},v_B,v_C))$. If $\lfloor r' \rfloor= \lfloor r \rfloor$, by definition $w_A(f_k(v_{A}',v_B,v_C))= w_A(f_k(v_{A},v_B,v_C))$ and we are done. Otherwise, we know that $\lfloor r_{v_{A}'} \rfloor > \lfloor r_{v_{A}} \rfloor$, so $\lfloor r_{v_{A}'} \rfloor \ge  \lfloor r_{v_{A}} \rfloor +1$. Therefore:
	\begin{equation*}
		w_{A}(f_k(v_{A}',v_{B},v_{C}))\underbrace{=}_{\substack{\text{$w_A$ is}\\ \text{the identity}\\\text{function}}}f_k(v_{A}',v_{B},v_{C})\ge \lfloor r_{v_{A}'} \rfloor \ge \lfloor r_{v_{A}} \rfloor +1 \ge f_k(v_{A},v_{B},v_{C})= w_{A}(f_k(v_{A},v_{B},v_{C}))
	\end{equation*}
\end{proof}
By Proposition \ref{myersonlemmainterval}, the fact that Alice's type space is an interval implies that the only normalized payment scheme that implements $f_k$ for Alice is:
\begin{equation} \label{alicepaymentdisjeq}
P_A (v_A,v_B,v_C)= r_{v_A}\cdot w_A(f_k(v_A,v_B,v_C)) - \int_{0}^{r_{v_{A}}}w_{A}(f_k(z\cdot w_A,v_{B},v_C))dz 
\end{equation}	
The following lemma is key in showing that $P_{f_{k}}$ is \textquote{small}. For all $v_B\in V_B$ and $v_C\in V_C$, denote $v_B^j$ and $v_C^j$ as the prefixes of length $j$ of $v_B$ and of $v_C$.
\begin{lemma} \label{helperlemma}
	Fix an alternative $j\in \mathcal{A}$, and an integer $t\in \{0,...,j\}$.
	Then, for all types $(v_B,v_C)\in V_B\times V_C$ where the intersection of $v_B^j$ and $v_C^j$ is of size $t$ and alternative $j$ is reachable from $(v_B,v_C)$, Alice's price for $j$ according to  the payment scheme $P_A$ is one of the following:
	\begin{enumerate}
		\item $j\cdot (j-1) - \frac{(j-1)\cdot (j-2)}{2} - t+1$. 
		\item $j^2 - \frac{j\cdot (j-1)}{2} - t$. 
	\end{enumerate}    
\end{lemma}
\begin{proof}
	Fix $t\in \{0,...,j\}$.
	Observe that there are \emph{at  most} two possible ways to reach alternative $j$. The first is when the $(j-1)$'th bits of $v_B$ and $v_C$ intersect and $\lfloor r_{v_A} \rfloor = j-1$.
	Denote the valuation $(j-1)\cdot w_A(\cdot)$ with $v_A$. In this case, the payment of Alice for all the types $v_B,v_C$ that satisfy those conditions is:
	\begin{align}
	P_A(v_A,v_B,v_C)  
%	(j-1)\cdot w_A(f_k(v_A,v_B,v_C)) - \int_{0}^{j-1}w_{A}(f(z,v_{B},v_C))dz  
	  &= (j-1)\cdot j - \int_{0}^{j-1}w_{A}(f_k(z\cdot w_A,v_{B},v_C))dz &\text{(plug in (\ref{alicepaymentdisjeq}))} \nonumber \\  
	&=j\cdot (j-1) 
	-\sum_{\substack{i=0 \\ \text{$i$'th bit} \\  \text{does not}\\ \text{ intersect}}}^{j-2} i  
	- \sum_{\substack{i=0 \\ \text{$i$'th bit} \\  \text{intersects}}}^{j-2} (i+1) \nonumber \\ 
	 &=j\cdot (j-1) -(\sum_{i=0}^{j-2}i) - (t-1) \label{tplugin}
	 \\ 
	 &= j\cdot (j-1) - \frac{(j-2)(j-1)}{2} - t + 1 \nonumber       
	\end{align}
	(\ref{tplugin}) holds because the $(j-1)$'th bits intersect and there are overall $t$ intersections in the $j$-prefixes of $v_B$ and of $v_C$. Thus, there are $t-1$ intersections in the $(j-1)$-prefixes, i.e., in indices $\{0,...,j-2\}$. 
	
	The second case is when the $j$'th bits of $v_B$ and $v_C$ do not intersect. This time, denote with $v_A$ the valuation $j\cdot w_A(\cdot)$. Similarly, the payment is:
\begin{align}
\label{disjpayment}
P_A(v_A,v_B,v_C) &=    
j^2 - \int_{0}^{j}w_{A}(f_k(z,v_{B},v_C))dz &\text{plug in (\ref{alicepaymentdisjeq})} \\
&=j^2 -\sum_{\substack{i=0 \\ \text{$i$'th bit} \\  \text{does not}\\ \text{ intersect}}}^{j-1} i  
- \sum_{\substack{i=0 \\ \text{$i$'th bit} \\  \text{intersects}}}^{j-1} (i+1) \nonumber \\  
 &= j^2 - \frac{j(j-1)}{2}-t \label{tplug2}
\end{align}
(\ref{tplug2}) holds because the $j$'th bits of $v_B$ and $v_C$ do not intersect and also do not belong in the prefixes of length $j$, $v_B^j$ and in $v_C^j$, so by definition there are $t$ intersections in indices $\{0,...,j-1\}$.  
\end{proof}
\begin{corollary} \label{okpaymentscorollary}
	$|\Ima P_A|= \mathcal{O}(k^2)$.
\end{corollary}	
\begin{proof}
Fix some alternative $j\in \mathcal{A}$. The size of the intersection of the first $j$ bits of $v_B$ and of $v_C$ is in the range $\{0,\ldots,j\}$, where $j\le k$. By Lemma \ref{helperlemma}, every such intersection size induces at most two possible prices for alternative $j$. Thus, the overall number of prices for alternative $j$ is at most $2k+2$. There are overall $k+1$ alternative, so indeed $|\Ima P_A|=\mathcal{O}(k^2)$. 
\end{proof}
\begin{lemma} \label{multiequalccicbig}
	Computing $P_A$ requires $\Omega(k)$ bits of communication. 
\end{lemma} 
\begin{proof}
We reduce from disjointness with $k-1$ bits. Let Bob's type be the input of the first player in the disjointness problem with extra zero bit at the end, and let Charlie's type be the input of the second player with an extra zero bit at the end. Let Alice's type be $v_A=(k-1)\cdot w_A$, i.e. $r_{v_{A}}=k-1$. Since the $(k-1)$'th bits are by construction not intersecting, alternative $k-1$ is chosen. If there are no intersecting bits in the $(k-1)$-disjointness problem, then by equation (\ref{disjpayment}), the payment is $(k-1)^2 - \frac{(k-1)(k-2)}{2}$. If there is an intersecting bit, then the payment is strictly smaller. Thus, computing the payment is at least as hard as computing disjointness with $k-1$ bits, and the lemma follows.       
\end{proof}
\begin{lemma} \label{lemmaallmechanismsdifficult}
	$2\cdot cc_{IC}(f_k)\ge cc(P_A)$. 
\end{lemma}
\begin{proof}
	Let $\Pi^{M^\ast}$ be the most efficient protocol that implements $f_k$, i.e., $cc_{IC}(f_k)=cc(M^{\ast})$. Fix the valuations $v_A,v_B,v_C\in V_A\times V_B \times V_C$.
	  For the payment of Alice, denote Alice's payment scheme according to $\Pi^{M^\ast}$ with $P^{\ast}_A$. Denote Alice's valuation when $r_A=0$ with $v_0$.
	  The uniqueness of $P_A$ as a normalized payment scheme allows us to use a similar argument to the proof of Claim \ref{ccicbiggerccaclaim}, and get that $P_A(v_A,v_B,v_C)=P_A^{\ast}(v_A,v_B,v_C)-P_A^{\ast}(v_0,v_B,v_C)$. Thus, $P_A$ can be computed using two calls to   $\Pi^{M^\ast}$.
\end{proof}

We now finish the proof of Theorem \ref{nopolylogthm}. $cc(f_k)\le \log(k)+2$, because $f_k$ can be easily computed by a protocol where Alice sends $\lfloor r_{v_A} \rfloor$ using $\log k$ bits, and Bob and Charlie send back their $\lfloor r_{v_A} \rfloor$'th bits. Also, there are $k$ alternatives, so by standard communication complexity arguments,  $cc(f)\ge \log k$.  Combining Lemmas \ref{multiequalccicbig} and \ref{lemmaallmechanismsdifficult}, we get that $cc_{IC}(f_k)=\Omega(k)$.  By Corollary \ref{okpaymentscorollary}, $|\Ima P_A|= \mathcal{O}(k^2)$, which completes the proof. 

\end{proof}

\subsection{Hardness of Welfare Maximizing Mechanisms}\label{vcgsection}
%	\colorbox{yellow}{\textbf{TODO - change it!! }}
In Subsections \ref{hardnessnumberofpaymentssec} and \ref{polylogsec}, we gave two examples for social choice functions for which there is an exponential gap between computing the output and payment computation. 
%One possible criticism about those results is that the functions that we construct are in fact quite contrived. 
In this section we show that there are natural social choice functions that exhibit this exponential gap. 

The natural social choice function that we construct is simply a welfare maximizer in a multi-unit auction with $m$ items. We use this multi-unit auction instance to show another gap: between the approximation ratio of truthful and non-truthful algorithms. For this instance, non-truthful algorithms achieve the optimal welfare  with $\mathcal{O}(\log m)$ communication (approximation ratio of $1$), whereas their truthful counterparts achieve the optimal welfare only if they use at least $\Omega(m)$ communication. It means, that all truthful algorithms with running time $o(m)$ have approximation ratio which is strictly less than $1$.  
 
% We now describe the construction, which is largely based on the domain and function described in Subsection \ref{hardnessnumberofpaymentssec}. 
 In a multi-unit auction with $m=2^k$ items, all items are identical and values of players are determined solely by the number of items they get: $v\colon \{0,\ldots,m\}\to \mathbb{R_+}$. The valuations are monotone ($l>j$ implies $v(l)\ge v(j)$) and normalized ($v(0)=0$).\footnote{Note that we use the term \textquote{normalized} to describe two different notions. The first one is a property of mechanisms (bidding zero guarantees a payment of zero) and the second one is a property of multi-unit valuations ($v(0)=0$).}  

%Fadel and Segal \cite{fs} observed that if the type space is small then the payment computation overhead is small as well. Necessarily, to prove the exponential gap the domain of the welfare maximizer that we construct is huge. As usual, the challenge is to guarantee that although the domain is huge, the computation of the output can be done efficiently.
\begin{theorem}\label{vcgtheorem}
For every integer $k\ge 1$,	there exists a multi-unit auction instance with $m=2^k$ items such that a welfare maximizing allocation can be computed with $\mathcal{O}(k)=\mathcal{O}(\log m)$ bits, but every truthful mechanism that maximizes the welfare requires at least $exp(k)=\Omega(m)$ communication. As a corollary, there exists a welfare-maximizing allocation rule $f$ such that $cc_{IC}(f)=\exp(cc(f))$.  
\end{theorem}
\begin{proof}
	The proof structure is as follows. First, we describe the multi-unit auction instance and prove that the valuations of players are monotone and normalized. We proceed by presenting a welfare-maximizing function $\hat{f}_k$ for it, which is based on the social choice function $f_k$ described in Section \ref{hardnessnumberofpaymentssec}. We show that $\hat{f}_k$ achieves the optimal welfare and can be computed using \textquote{few} bits (Claim \ref{welfaremaximizerclaim}).
	The next step is showing hardness of truthful algorithms for this instance. There might be more than one social choice function that maximizes the welfare, where functions differ by their tie-breaking rule, i.e. by which alternative is chosen whenever more than one alternative yields the maximum welfare. Thus, proving that implementing $\hat{f}_k$ is \textquote{hard} is not enough. A stronger claim is needed: that \emph{every} truthful welfare-maximizing mechanism for this instance requires $\exp(k)=\Omega(m)$ bits (Claim \ref{hardnessofmech}).

	 The instance is as follows. It is based on the construction of $f_k:V_A\times V_B\to \mathcal{A}$ in Subsection \ref{hardnessnumberofpaymentssec}.  The set of alternatives is $\mathcal{A}=\{a_0,\ldots,a_{2^k}\}$ for this instance as well. Alternative $a_i$ now stands for an allocation where Alice wins $i$ items and Bob wins $m-i$ items, where $m=2^{k}$. 
	We now describe the domains of the players, $\hat{V}_A$ and $\hat{V}_B$. Alice's domain in the instance of Subsection \ref{hardnessnumberofpaymentssec} is $V_A=\{r\cdot w_A(\cdot) \hspace{0.25em}|\hspace{0.25em}r\in [0,2^{k+1}-1] \}$ where $w_A(a_i)=|\mathcal{A}|^{4ik}-1$. For each $v_{A}\in V_{A}$, we include  the valuation $\hat{v}_A$ where $\hat{v}_A(i)=r_{v_A}\cdot w_A(a_i)$. Observe that for each $i$ we have that $v_A(a_i)=\hat{v}_A(i)$.

	Similarly, we define a valuation $\hat v_B\in \hat V_B$ for each $v_B\in V_B$. $\hat v_B$ relies on $P_A(\cdot,v_B)$, i.e., on Alice's payment function  induced by the normalized implementation of the function $f_k$. In order to define the valuation $\hat{v}_{B}$ that $v_B$ induces, we reformulate $P_A(\cdot,v_B)$ to be a function of the alternative chosen and of $v_B$ (using the taxation principle).  Note that $f_k$ is player decisive for Alice, because for every alternative $a_i$ and every $v_B\in V_B$, there is some  $v_A\in V_A$ such that $f_k(v_A,v_B)=a_i$.\footnote{\label{footnote}We remind that $v_B$  partitions the integers in $V_A$ to alternatives in a way that guarantees that every alternative has some $v_A$ that reaches it. This is due to the fact that in Section \ref{hardnessnumberofpaymentssec}, we defined the set $\mathcal{C}$ such that every $c\in \mathcal{C}$ satisfies  $c_i\ge 1$ for all $i\in \{0,\ldots,|\mathcal{A}|-1\}$.}	Hence,  $P_A(\cdot,v_B)$ is well defined for all $a_i\in \mathcal{A}$. 
	Fix $v_B\in V_B$. Include in $\hat{V}_B$ the valuation $\hat{v}_B$ which is defined as follows: for every $i\in\{0,\ldots,m\}$, $\hat{v}_B(i)=P_A(a_m,v_B)- P_A(a_{m-i},v_B)$. Observe that $v_A\to \hat{v}_A$ and $v_B\to \hat{v}_B$ are bijections. 
	We conclude the construction of the domains by adding another valuation to $\hat{V}_A$,  $\hat{v}_A^\ast:\{0,\ldots,m\}\to\mathbb{R}$. 
%	It is defined by:
	\begin{equation} \label{specialvaluationeq}
	\hat{v}_A^\ast(i)=\begin{cases}
	\max_{\hat{v}_B\in \hat{V}_B }\{\hat{v}_B(m)\} + 1 &\qquad i=m \\
	0 &\qquad i\in\{0,\ldots,m-1\}
	\end{cases}
	\end{equation}
As we will prove later on, we add this valuation to $\hat{V}_A$ because if Alice bids it, she \textquote{forces}  every optimal algorithm for this auction to allocate all $m$ items to her, no matter what is  the valuation of Bob or the tie-breaking rule of the algorithm.   

We now explain why all the valuations in $\hat{V}_A$ and in $\hat{V}_B$ are monotone and normalized. It is true for $\hat{V}_A$ because $w_A$ is monotone in the index of the alternative, and $\hat{v}_A(0)=v_A(a_0)=r_{v_A}\cdot w_A(a_0)=0$. Also, $\hat{v}_A^\ast$ is monotone and normalized. Regarding Bob, fix a type $\hat{v}_B\in \hat{V}_B$. By construction, $\hat{v}_{B}(0)=P_A(a_m,v_B)-P_A(a_m,v_B)=0$. As for monotonicity, we show that for all $j\in \{0,\ldots,m-1\}$: 
\begin{align} \label{truthfuleqvcg}
\hat{v}_B(j+1) > \hat{v}_B(j) \iff
P_A(a_m,v_B)- P_A(a_{m-j-1},v_B) &>
P_A(a_m,v_B)- P_A(a_{m-j},v_B)     \\ \iff 
P_A(a_{m-j},v_B) &> P_A(a_{m-j-1},v_B) \nonumber 
\end{align}
Recall that $P_A$ implements $f_k$ which is player decisive for Alice, and that every valuation in $V_A$ is strictly monotone in the index of the alternative chosen. Therefore, the truthfulness of $P_A$ implies that $P_A(\cdot,v_B)$ must be strictly monotone in the index of the alternative, so $P_A(a_{m-j},v_B) > P_A(a_{m-j-1},v_B)$ holds.

In order to prove the theorem, we need the following technical lemma. 
	  	   \begin{lemma}\label{technicallemma}
	  	   	Every welfare maximizing function $f:\hat{V}_A\times \hat{V}_B\to \mathcal{A}$ satisfies that:
	  	   	\begin{enumerate}
	  	   		\item \label{case1} For all $\hat{v}_B\in \hat{v}_B$, $\hat{v}_A\equiv 0$ implies that  $f(\hat{v}_A,\hat{v}_B)=a_0$.
	  	   		\item \label{case2} For all $\hat{v}_B\in \hat{v}_B$, $f(\hat{v}_A^\ast,\hat{v}_B)=a_m$. (where $\hat{v}_A^\ast$ is the valuation described in (\ref{specialvaluationeq}))  \end{enumerate}
	  	   	In words, every welfare maximizing function satisfies that Alice gets no items at all if she bids her zero valuation, and she gets all $m$ items if she bids $\hat{v}_A^\ast$.
%	  	   	In words, every welfare maximizing function satisfies that Alice gets all $m$ items if she bids $\hat{v}_A^\ast$, and she gets no items at all if she bids her zero valuation.	
  	   		\end{lemma}
  	   		\begin{proof}
  	   			We want to show that if $\hat{v}_A\equiv 0$, then $a_0$ is the \emph{only} alternative that maximizes the welfare, and that the same holds for $\hat{v}_A^\ast$ and $a_m$. 
%		Recall that all the valuations of Bob are strictly increasing in the number of items .

		For all $\hat{v}_B\in \hat{V}_B$, if Alice bids the constant valuation $\hat{v}_A\equiv 0$, then her value is  unaffected by the number of items that she gets, whilst Bob's valuation is strictly increasing in the number of items (see (\ref{truthfuleqvcg})). Hence, the optimal welfare is achieved \emph{solely} by allocating all items to Bob, i.e. outputting $a_0$. 
	
		As for the other case,  if Alice bids $\hat{v}_A^\ast$ and Bob bids some $\hat{v}_B\in \hat{V}_B$, the welfare that alternative $a_m$ achieves is:
		\begin{equation}\label{mwelfare}
		\hat{v}_A^\ast(m)+ \underbrace{\hat{v}_B(0)}_{=0} \underbrace{=}_\text{by (\ref{specialvaluationeq})} \max_{\hat{v}_B\in \hat{V}_B }\{\hat{v}_B(m)\} + 1  
		\end{equation}
		For every $i\in \{0,\ldots,m-1\}$, the welfare that $a_i$ obtains is:
		\begin{equation}\label{iwelfare}
		\underbrace{\hat{v}_A^\ast(i)}_\text{$=0$, by (\ref{specialvaluationeq})} + \hat{v}_B(m-i) \le \hat{v}_B(m) < \max_{\hat{v}_B\in \hat{V}_B }\{\hat{v}_B(m)\} + 1    
		\end{equation}
		Combining (\ref{mwelfare}) and (\ref{iwelfare}), we get that  $a_m$ is the \emph{only} alternative that achieves the maximal welfare.
  	   		\end{proof}
		\begin{claim} \label{welfaremaximizerclaim}
			There exists a social choice function $\hat{f}_k:\hat{V}_A\times \hat{V}_B\to\mathcal{A}$ such that:
			\begin{enumerate}
				\item $\hat{f_k}$ is a welfare maximizer, i.e. for every $\hat{v}_{A}\in \hat{V}_{A}, \hat{v}_{B}\in \hat{V}_{B}$:
				\begin{equation*}
				\hat{f}_k(\hat{v}_{A},\hat{v}_{B}) \in \argmax_{i\in \{0,1,\ldots,m\}} \hat{v}_{A}(i) +  \hat{v}_{B}(m-i)    
				\end{equation*}
				\item $cc(\hat{f}_k)=\mathcal{O}(k)$. 
			\end{enumerate}
		\end{claim}
	 
	\begin{proof} 
We define $\hat{f}_k$ as follows:
\begin{equation*} 
	\hat{f}_k(\hat{v}_A,\hat{v}_B) = \begin{cases}
	a_{m} &\qquad \hat{v}_A= \hat{v}_A^\ast \\
	f_k(v_A,v_B) &\qquad \text{otherwise.}
	\end{cases}
\end{equation*}
where $f_k$ is the social choice function described in Subsection \ref{hardnessnumberofpaymentssec}, $\hat{v}_A$ is induced by $v_A$ and $\hat{v}_B$ is induced by $v_B$. We first explain why $\hat{f}_k$ is a welfare maximizer.   

Fix $\hat{v}_{A} \in \hat{V}_{A}, \hat{v}_{B}\in \hat{V}_{B}$. By Lemma \ref{technicallemma}, if $\hat{v}_A=\hat{v}_A^\ast$,  outputting $a_m$ indeed achieves the optimal welfare. 
	If $\hat{v}_A\neq \hat{v}_A^\ast$, then we denote  $\hat{f}_k(\hat{v}_A,\hat{v}_B)=f_k(v_A,v_B)$ with $a_j$, i.e. Alice wins $j$ items and Bob wins $m-j$ items. We want to prove that for all $i\in \{0,1,\ldots,m\}$:
	\begin{equation*}
	\hat{v}_{A}(j)  + \hat{v}_{B}(m-j)\ge \hat{v}_{A}(i)  + \hat{v}_{B}(m-i)     
	\end{equation*}
	   
	Recall that $P_A$ is the  payment function of Alice in the normalized implementation of $f_k$, which is player decisive for Alice (see footnote \ref{footnote}). Hence, the fact that $P_A$ implements $f_k$ truthfully implies that for all $i\in \{0,\ldots,m\}$:
		\begin{equation} \label{truthfuleq}
		f_k(v_A,v_B)=a_j\implies v_{A}(a_j)- P_A(a_j,v_{B}) \ge v_{A}(a_i)- P_A(a_i,v_{B})
		\end{equation}
		 By the  construction of  $\hat{v}_B$, $P_A(a_j,v_B)=
		 P_A(a_m,v_B)-\hat{v}_B(m-j)$ and $P_A(a_i,v_B)=
		 P_A(a_m,v_B)-\hat{v}_B(m-i)$. We substitute those equalities into (\ref{truthfuleq}):
		 \begin{equation*}
		 v_{A}(a_j)- P_A(a_m,v_B)+\hat{v}_B(m-j)  \ge v_{A}(a_i)- P_A(a_m,v_B)+\hat{v}_B(m-i)
		 \end{equation*}
		Recall that $v_A(a_j)=\hat{v}_A(j)$:
		\begin{equation*}
		\hat{v}_A(j)+\hat{v}_B(m-j)  \ge \hat{v}_A(i)+\hat{v}_B(m-i)
		\end{equation*}
		
%		
%			\begin{equation*}
%			\hat{v}_{A}(j) + \hat{v}_{B}(m-j) \ge \hat{v}_{A}(l) + \hat{v}_{B}(m-l)  		\underbrace{\ge}_{\substack{{\text{$\hat{v}_{A}$ is  }}\\{\text{monotone }}}} 
%			\hat{v}_{A}(i) + \hat{v}_{B}(m-l) \underbrace{=}_{\hat{v}_{B}(m-l)=\hat{v}_{B}(m-i)} \hat{v}_{A}(i) + \hat{v}_{B}(m-i)      
%			\end{equation*}
		
which is the desired conclusion.
We now show that $cc(\hat{f}_k)=\mathcal{O}(k)$. Consider the following protocol: Alice sends a bit that specifies whether her valuation is $\hat{v}_A^\ast$ or not. If it is $\hat{v}_A^\ast$, both players know that the output is $a_m$. Otherwise, Alice and Bob translate their  multi-unit valuations $\hat{v}_A$,$\hat{v}_B$ to the valuations that induce them, $v_A$ and $v_B$, with no communication. Then, they execute the protocol of $f_k$ with $v_A,v_B$ and output the alternative.
By Theorem \ref{mostimportanttheorem}, $cc(f_k)=\mathcal{O}(k)$, so the total communication of the protocol is $\mathcal{O}(k)$ as well. 
%By Theorem \ref{mostimportanttheorem}, $cc(f_k)=\mathcal{O}(k)$, so $cc(\hat{f}_k)=\mathcal{O}(k)$ as well.     

%Notice that $\hat{f}_k(\hat{v}_{A},\hat{v}_{B})=f_k(v_A,v_{B})$ for some $v_A,v_B$, and that the conversion of $\hat{v}_A$ to $v_A$ and of $\hat{v}_B$ to $v_B$ can be done by Alice and Bob separately, with no communication involved. Thus, $f_k$ can be used to compute $\hat{f}_k$, so:
%\begin{equation*}
%cc(\hat{f}_k)\le cc(f_k)\underbrace{=}_\text{by Theorem \ref{mostimportanttheorem}}\mathcal{O}(k) 
%\end{equation*}
	\end{proof}
Thus, a welfare-maximizing allocation can be computed non-truthfully with $\mathcal{O}(k)=\mathcal{O}(\log m)$ bits. It remains to show that achieving this optimum truthfully requires $exp(k-1)=\Omega(m)$ bits. 
\begin{claim} \label{hardnessofmech}
	Let $f_k^\ast:\hat{V}_A\times \hat{V}_B\to \mathcal{A}$ be a welfare-maximizing function. Then, every mechanism $M^\ast$ that implements $f_k^\ast$ requires $\Omega(2^{k})$ bits.  
\end{claim}
\begin{proof}[Proof of Claim \ref{hardnessofmech}]
	Fix a social choice function $f_k^\ast$ and a truthful mechanism for it $M^\ast=(f_k^\ast,P_A^\ast,P_B^\ast)$.
	We say that an auction mechanism is normalized if bidding a constant valuation of $\vec{0}$ guarantees a payment of $0$.\footnote{Similarly to the definition of normalization for single parameter domains in Section \ref{setupmyersonsection}.}
	We explain why we can assume without loss of generality that $M^\ast$ is normalized.
	All the valuations of Bob in $\hat{V}_B$ are strictly increasing, so $\vec{0}\notin \hat{V}_B$, and $P_B^\ast$ is trivially normalized. Regarding $P_A^\ast$, observe that 
	$P_A^\ast(\hat{v}_A,\hat{v}_B)-P_A^\ast(\vec{0},\hat{v}_B)$ is normalized and truthful (due to the same arguments applied in Claim \ref{ccicbiggerccaclaim}).
	Thus, a normalized mechanism can be obtained by executing twice the protocol of $M^\ast$, so it suffices to prove that $cc(M^\ast)= \Omega(2^k)$ where $M^\ast$ is normalized. 
	
	We do so by showing that $cc(P_A)\le cc(M^\ast)+cc(f_k)$, where $f_k$ is the function described in Subsection \ref{hardnessnumberofpaymentssec}, and $P_A$ is Alice's payment according to its normalized implementation. Since $cc(P_A)\ge 2^k$ (by Corollary \ref{ppaboundcorollary}) and $cc(f_k)=\mathcal{O}(k)$ (by Theorem \ref{mostimportanttheorem}), $cc(P_A)\le cc(M^\ast)+cc(f_k)$ implies that $cc(M^\ast)= \Omega(2^k)$. 
	Intuitively, we prove that $cc(P_A)\le cc(M^\ast)+cc(f_k)$ by
	showing that $M^\ast$ can be used to extract Bob's type, which is  \textquote{almost} sufficient
	for the computation of $P_A(v_A,v_B)$.   
	
%	Intuitively, we show that $M^\ast$ can be used to compute $P_A$ by  is hard by showing how the protocol of the mechanism $M^\ast$ can be used in order to extract Bob's type, which we show to be  \textquote{almost} sufficient
%for the computation of $P_A:V_A\times V_B\to\mathbb{R}$. 

To this end, we explain why  $f_k^\ast:\hat{V}_A\times \hat{V}_B\to \mathcal{A}$
	 has a \emph{unique} normalized payment scheme. Note that $f_k^\ast$ is implementable because it is a welfare maximizing function. In particular, if we restrict $f_k^\ast$ to the domain $\hat{V}_A/\{\hat{v}_A^\ast\}$, it remains implementable. 
	We remind that the domain $\hat{V}_A/\{\hat{v}_A^\ast\}$ contains the same valuations as the domain $V_A =\{V_A=\{r\cdot w_A(\cdot) \hspace{0.25em}|\hspace{0.25em}r\in [0,2^{k+1}-1]\}$ (because for all $i\in\{0,\ldots,m\}$, $\hat{v}_A(i)=v_A(a_i)$). Thus, we can apply Myerson's Lemma (Proposition \ref{myersonlemmainterval}) on $f^\ast_k:\hat{V}_A/\{\hat{v}_A^\ast\}\times \hat{V}_B\to \mathcal{A}$ and get that $f_k^\ast$ has a unique normalized payment scheme for Alice, whenever the valuations of Alice and Bob are drawn from the domains   $\hat{V}_A/\{\hat{v}_A^\ast\}$ and $\hat{V}_B$. It implies that $f_k^\ast$ has a unique normalized payment scheme also after including  $\hat{v}_A^\ast$ in Alice's domain.\footnote{Because adding another valuation only
		restricts the set of truthful payment schemes, so it cannot increase their number.  Hence, the facts that $f_k^\ast$ satisfies uniqueness of payments for $\hat{V}_A/\{\hat{v}_A^\ast\}\times \hat{V}_B$ and that $f_k^\ast$ is implementable for $\hat{V}_A\times \hat{V}_B$ jointly imply  that $f_k^\ast:\hat{V}_A\times \hat{V}_B\to\mathcal{A}$ satisfies uniqueness of payments as well.} 

%Also, observe that any truthful payment scheme $P$ for Alice can be transformed to a  normalized and truthful payment scheme by executing it twice with inputs $(\hat{v}_A,\hat{v}_B)$ and  $(0,\hat{v}_B)$ and outputting $P(\hat{v}_A,\hat{v}_B)-P(0,\hat{v}_B)$. $P(\hat{v}_A,\hat{v}_B)-P(0,\hat{v}_B)$ is equal to $0$ whenever $\hat{v}_A\equiv 0$ (so it is normalized), and it is also truthful because $P$ is truthful, and translation of payment schemes preserves truthfulness. Thus, if all mechanisms for $f_k^\ast$ where Alice's payment scheme is normalized require $\exp(k)$ bits, it implies that \emph{all} mechanisms for $f_k^\ast$ require $\exp(k-1)$ bits. Hence, we assume without loss of generality that $P_A^\ast$ is normalized. 

Recall that $f_k^\ast$ maximizes the welfare, so it is well known to be implemented by the payment scheme $P_{vcg}(j,\hat{v}_B)=\hat{v}_B(m)-\hat{v}_B(m-j)$.
%\footnote{It is easy to see that $P_{vcg}$ corresponds to  a payment function in the VCG mechanism.}
 We now show that $P_{vcg}$ is normalized, which means that  $P_{vcg}\equiv P_A^\ast$, because $f_k^\ast$ has a unique normalized payment. 
By Lemma \ref{technicallemma}, if Alice bids $\hat{v}_A\equiv 0$, then $f_k^\ast$ necessarily outputs $a_0$, which means that Alice wins $0$ items.   
Therefore, her payment according to $P_{vcg}$ is $P_{vcg}(0,\hat{v}_B)=\hat{v}_B(m)-\hat{v}_B(m)=0$, so $P_{vcg}$ is normalized. Hence,  $P_A^\ast(a_j,\hat{v}_B)=P_{vcg}(j,\hat{v}_B)=\hat{v}_B(m)-\hat{v}_B(m-j)$. Thus, for $j=m$:
\begin{equation*}
P_A^\ast(a_m,\hat{v}_B)= \hat{v}_B(m)- \underbrace{\hat{v}_B(0)}_{=0} \underbrace{=}_{\substack{\text{by construction}\\\text{of $\hat{v}_B$}}} P_A(a_m,v_B)-P_A(a_0,v_B)  
\end{equation*}
We now explain why $P_A(a_0,v_B)=0$. Recall that $f_k$ is player decisive for Alice and monotone,
so bidding zero ($r_{v_A}=0$) guarantees that $a_0$ is chosen. Also, according to the formula of $P_A$   (specified in (\ref{alicepayments})), we get that $P_A(v_A,v_B)=0$ if $r_{v_{A}}=0$. Thus, a zero bid implies both that $a_0$ is chosen and that $P_A(v_A,v_B)=0$.
Hence, by the taxation principle, $P_A(a_0,v_B)=0$. 
We also remind that $m=2^k$. We get that:
\begin{equation}\label{payment}
P_A^\ast(a_m,\hat{v}_B)= P_A(a_m,v_B)- \underbrace{P_A(a_0,v_B)}_{=0}= P_A(a_{2^k},v_B)
\end{equation}

We now a describe a protocol for $P_A$ with $cc(M^\ast)+cc(f_k)$ bits. Let $\Pi^{M^\ast}$ and $\Pi^{f_k}$  be protocols of $M^\ast$ and of $f_k$, respectively. 

Fix $v_A,v_B$. Execute  $\Pi^{M^\ast}(\hat{v}_A^\ast,\hat{v}_B)$, where $v_A^\ast$ is the valuation specified in (\ref{specialvaluationeq}) and $\hat{v}_B$ is the valuation induced by $v_B$. $M^\ast$ maximizes the welfare, so by Lemma \ref{technicallemma},  $f_k^\ast(\hat{v}_A^\ast,\hat{v}_B)=a_m$. 
Thus, by executing $M^\ast(\hat{v}_A^\ast,\hat{v}_B)$, we get that alternative $a_m$ is chosen, so by equation (\ref{payment}), Alice pays   $P_A(a_{2^k},v_B)$. 
By Claim \ref{ppaboundcorollary}, each $v_B\in V_B$ induces a different value of $P_A(a_{2^k},v_B)$, so once Alice knows $P_A(a_{2^k},v_B)$, she knows Bob's valuation $v_B$ completely.
Next, Alice and Bob execute $\Pi^{f_k}(v_A,v_B)$.  
Afterwards, they both know  $f_k(v_A,v_B)$ and Bob's valuation. Thus, according to the taxation principle, they both know $P_A(f_k(v_A,v_B),v_B)=P_A(v_A,v_B)$, which completes the proof.     
\end{proof}
Note that Claim \ref{hardnessofmech} implies the corollary of Theorem \ref{vcgtheorem}.
All welfare-maximizing truthful mechanisms for this instance require $\Omega(2^k)$ bits, so in particular every mechanism that implements $\hat{f}_k$ requires $\Omega(2^k)$ bits. Thus, $cc_{IC}(\hat{f}_k)= \Omega(2^{k}) \approx \exp(cc(\hat{f}_k))$. 

\end{proof}

%\section{Algorithms for Payment Computation} \label{algorithmssection}

\section{Truthful in Expectation Mechanisms} \label{TIE-appendix-section}
Up until now we showed that the communication cost of ex-post implementations of social choice functions might be exponential comparing to output computation. However, we observe that in many domains the payments of the (deterministic) social choice function can be computed randomly so that the expected value of the payment equals the value of the (deterministic) ex-post payment scheme. If the players are risk neutral, this gives us a truthful-in-expectation implementation of the social choice function. For all the domains below we prove that $cc_{TIE}(f)\le poly(n,cc(f))$.

For single parameter domains, where the private information $R_i$ of each player is either an interval $[0,b_i]$ or a finite set with non-negative values, the computation of payment is based on the observation that to compute the expected value of the integral in Myerson's payment formula, it suffices to evaluate the value of the integral at a random valuation and know the type of the player $i$, $v_i$. One point of potential complication is that representing the type of player $i$ might be much more costly than computing $cc(f)$. We get around this problem by  essentially providing a ``similar'' type to $v_i$, which is based on the communication protocol.

We then extend our results to some multi-parameter domains. We first consider scalable domains: domains where for each constant $\lambda\in [0,1]$, if the type $v(\cdot)$ is in the domain, then so does $\lambda\cdot v(\cdot)$. We rely on a result of \cite{babaioff2015truthful} who show an integral-based payments formula similar to Myerson's for this domain. We rely on this formula in the sense that we compute a payment which is equal in expectation to it, similarly to the single parameter case. Finally, we show that scalable and convex domains are computationally equivalent.  
%Exact statements of the results and proofs can be found in Section \ref{TIE-appendix-section}.
%We now show that (deterministic) social choice functions can be efficiently implemented by truthful-in-expectation mechanisms for several settings: single parameter, scalable and convex. 
We begin by formally defining truthfulness in expectation.  
\begin{definition}
	Let $f\colon V_1\times \cdots\times V_n\to \mathcal{A}$ be a (deterministic) social choice function. A mechanism $M=(f,P)$ is truthful in expectation if for every player $i$, every $v_{-i}\in V_{-i}$ and every $v_i,v_i'\in V_i$:
	\begin{equation} \label{expectedrequirement}
	\mathbb{E}[v_i(f(v_i,v_{-i}))-P_i(v_i,v_{-i})] \ge \mathbb{E}[v_i(f(v'_i,v_{-i}))-P_i(v_i',v_{-i})] 
	\end{equation} 
	where the expectation is taken over the randomness of $P$. 
\end{definition}

We denote with $cc_{TIE}(f)$ the communication complexity of the optimal truthful-in-expectation implementation for $f$.

\subsection{Single Parameter Domains}
We are especially interested in truthful in expectation implementation for single parameter mechanisms, because they demonstrate an exponential gap between the communication complexities of deterministic truthfulness and  truthfulness in expectation. The gap is established by observing that the functions used in Section \ref{lowerboundsec} to derive lower bounds satisfy $cc_{IC}(f)=\exp(cc(f))$ and $cc_{TIE}(f)\le poly(cc(f))$, where the latter statement follows from Theorem \ref{truthfulinexpectationthm}. 
\begin{theorem} \label{truthfulinexpectationthm}
	Let $f\colon V= V_1\times \cdots\times V_n\to \mathcal{A}$ be an implementable social choice function over a single parameter domain where for every player $i$, $R_i$ is an interval $[0,b_i]$ such that $b_i\in \mathbb{R}$. 
	Then:
	\begin{equation*}
	cc_{TIE}(f)\le (n+1)\cdot  cc(f)
	\end{equation*}
	As a corollary, if $R_i$ is a finite set with non-negative values for every player $i$, $cc_{TIE}(f)\le (n+1)\cdot  cc(f)$.  
\end{theorem}
For the proof of Theorem \ref{truthfulinexpectationthm}, 
we obtain an unbiased estimator for an integral using uniform sampling, similarly to \cite{archer2004approximate} and  \cite{babaioff2015truthful}. Let $U[a,b]$ be the continuous uniform distribution over the interval $[a,b]$.
\begin{lemma} \label{sampleintegrallemma}
	Let $g:[a,b]\to \mathbb{R}$ be an integrable function. Define a random variable $R=(b-a)\cdot g(z)$, where $z\sim U[a,b]$. Then, $R$ is an unbiased estimator of $\int_{a}^{b}g(x)dx$. 
\end{lemma}
\begin{proof}
	\begin{align*}
	\mathbb{E}_z[(b-a)\cdot g(z)] &= (b-a) \cdot \mathbb{E}_z[g(z)] \\
	&= (b-a) \cdot \int_{-\infty}^{\infty}g(z)\cdot f_Z(z) dz & \text{(law of the unconcious statistician)} \\
	&= (b-a) \cdot \int_{a}^{b}g(z)\cdot \frac{1}{b-a} dz  & \text{($f_Z(z)=\frac{1}{b-a}$ for $z\in [a,b]$, and $0$ otherwise)} \\
	&= \int_{a}^{b}g(z) dz      
	\end{align*}
\end{proof}
\begin{proof}[Proof of Theorem \ref{truthfulinexpectationthm}]
	Denote a protocol of $f$ with $\pi$. We begin by running $\pi(v_1,..,v_n)$. Fix a player $i$. $V_i$ is an interval, so Proposition \ref{myersonlemmainterval} yields that $f$ is \emph{deterministically} implemented by: 
	\begin{equation} \label{myersonpayments}
	P_{i}(v_{i},v_{-i})=r_{v_{i}}\cdot w_{i}(f(v_i,v_{-i})) - \int_{0}^{r_{v_{i}}}w_{i}(f(z\cdot w_i,v_{-i}))dz
	\end{equation}
	Hence, in order to obtain a payment which is truthful in expectation it suffices to compute a payment scheme whose \emph{expected value}  is (\ref{myersonpayments}). 
	Each leaf in the protocol $\pi$ is a combinatorial rectangle, $L=L_1\times\cdots\times L_n$. 
	For each leaf $L$, the players agree in advance on a profile which belongs in the leaf, i.e.  $(v_1^L,..,v_n^L)\in L$. Denote the leaf that $(v_1,..,v_n)$ reaches with $L^\ast$, and denote its agreed upon type for player $i$ with $v_i^\ast$. 
	\begin{lemma} \label{samepaymentlemma}
		$P_i(v_i^\ast,v_{-i})=P_i(v_i,v_{-i})$.
	\end{lemma}
	\begin{proof}
		By definition,
		$(v_i,v_{-i})\in L^\ast$ and $v_i^\ast \in L^\ast_i$, so by the mixing property $(v_i^\ast,v_{-i})\in L^\ast$. Thus, $w_i(f(v_i^\ast,v_{-i}))=w_i(f(v_i,v_{-i}))$, so from truthfulness we get that $P(v_i^\ast,v_{-i})=P(v_i,v_{-i})$.
	\end{proof}
	Hence, it suffices to compute a random variable with expectation $P_i(v_i^\ast,v_{-i})$. Note that $v_i^\ast$ is known to all players, due to the execution of $\pi$. Some player $j\neq i$ samples $z\sim U[0,r_{v_{i}^\ast}]$. Note that if we let player $i$ sample $z$ he could have potentially misreport the sample in order to increase the profit. All players in $N/\{i\}$ simulate $\pi(z\cdot w_i,v_{-i})$ in order to obtain $f(z\cdot w_i,v_{-i})$.  The output is:
	\begin{equation} \label{psample}
	\hat{P}=r_{v_{i}^\ast}\cdot w_i(f(v_i^{\ast},v_{-i}))- 
	r_{v_{i}^\ast} \cdot 
	w_i(f(z\cdot w_i,v_{-i}))  
	\end{equation}
	By Lemma \ref{sampleintegrallemma},  $\mathbb{E}[\hat{P}]=P_i(v_i^\ast,v_{-i})$. Due to the two executions of $f$, the players know all the components of $\hat{P}$. Hence, we obtained a truthful in expectation payment for player $i$ by making one extra call to $\pi$. Thus, truthful in expectation implementation of $f$ requires at most $(n+1)\cdot cc(f)$ bits. 
	
	For the corollary, let $R_i$ be a finite domain with non-negative values. Fix a player $i$ and a type $(v_1,\ldots,v_n)$. The players simulate $\pi(v_1,\ldots,v_n)$. Let $[0,b_i]$ be an interval that contains all the elements in $R_i$. 
	We extend $f$ to output for every $x \notin R_i$ the same alternative it assigns to the nearest $r_i$ which is smaller than $x$. For all $x$ smaller than $\min r_i$, the extension $f_{ext}$ always outputs an arbitrary alternative $a$ such that $w_i(a)$ is minimal. This extension preserves the monotonicity of $f$ and its new domain is an interval, so by Proposition \ref{myersonlemmainterval}, $f_{ext}$ is \emph{deterministically} implemented by: 
	\begin{equation*} \label{payments2}
	P_{i}(v_{i},v_{-i})=r_{v_{i}}\cdot w_{i}(f_{ext}(v_i,v_{-i})) - \int_{0}^{r_{v_{i}}}w_{i}(f_{ext}(z\cdot w_i,v_{-i}))dz
	\end{equation*}
	$f_{ext}(v)=f(v)$ for all $v\in V$, so 
	the truthfulness of the payment scheme for $f_{ext}$ implies truthfulness for $f$. Also, notice that $\pi$, the protocol of $f$, computes not only $f$, but also its extension. Thus, one extra call to $\pi$ is needed for the computation of a truthful in expectation payment for player $i$, so $cc_{TIE}(f)\le (n+1)\cdot cc(f)$. 
\end{proof}
\subsection{Scalable Domains}
Roughly speaking, scalable domains are multi-parameter domains that can be \textquote{stretched}. They are useful because of two main properties. The first is that a scalable domain can be projected to a single parameter domain, so upper bounds of payment computation in single parameter settings extend to them. The latter is that they are essentially equivalent (up to translation) to convex domains, so we use them as a means to derive upper bounds for them (as we define and prove formally in Subsection \ref{convexexpectationsec}). Formally, scalable domains are:   
\begin{definition}
	A domain of a player $V_i$ is \emph{scalable} if for every $v_i\in V_i$ and every $\lambda \in [0,1]$, $\lambda \cdot v_i \in V_i$. 
\end{definition}
By definition, a scalable domains necessarily contains a zero valuation,  $v_i\equiv \vec{0}$. Thus, for scalable domains, we say that  a mechanism is \emph{normalized} if $P_i(\vec{0},v_{-i})=0$ for all $v_{-i}$. 
As observed by \cite{babaioff2015truthful}, a corollary of Rochet \cite{rochet1987necessary} is:
\begin{proposition}[\cite{babaioff2015truthful,rochet1987necessary}] \label{scalableuniquenessthm} 
	Let $f\colon V_1\times \cdots \times V_n$ be an implementable social choice function with scalable domains. Then, the mechanism $(f,P)$ is truthful and normalized if and only if for every player $i$:
	\begin{equation} \label{scalablepayment}
	P_i(v_i,v_{-i})= v_i(f(v))-\int_{0}^{1}v_i(f(t\cdot v_i,v_{-i}))dt
	\end{equation}
\end{proposition} 
\begin{theorem} \label{scalableinexpectationthm}
	Let $f\colon V_1\times \cdots\times V_n\to \mathcal{A}$ be a social choice function with scalable domains. Then, $cc_{TIE}(f)\le (n+1)\cdot cc(f)$.
\end{theorem}
\begin{proof}
	Let $\pi$ be a communication protocol for $f$. The players simulate $\pi (v_1,\ldots,v_n)$. By Proposition \ref{scalableuniquenessthm}, the payment in (\ref{scalablepayment}) deterministically implements $f$, so for every player we wish to compute a random variable whose expected value is equal to it.
	For every leaf in $\pi$, the players agree in advance on a profile which belongs in the leaf. Denote the leaf that $\pi (v_1,\ldots,v_n)$ reaches with $L^\ast$, and its agreed type with $v_i^\ast$. Lemma \ref{samepaymentlemma} allows us to focus on computing a random variable with expectation $P_i(v_i^\ast,v_{-i})$. We obtain such a random variable when a player $j\neq i$ samples $t\sim U[0,1]$. The players simulate $\pi(t\cdot v_i^\ast,v_{-i})$ and output:
	\begin{equation*} 
	\hat{P}_i = v^\ast_i(f(v^\ast_i,v_{-i}))  - v_i^\ast(f(t\cdot v_i^\ast,v_{-i}))
	\end{equation*}
	By Lemma \ref{sampleintegrallemma}, $\mathbb{E}[\hat{P}_i]=P_i(v_i^\ast,v_{-i})$. Due to the two executions of $\pi$, the players know all the components of $\hat{P}$. By repeating for all players, we get a truthful-in-expectation implementation of $f$ with $(n+1)\cdot cc(f)$ bits. 
\end{proof}
\subsection{Convex Domains} \label{convexexpectationsec}
Convex domains are useful for mechanism design since they are weak monotonicity domains \cite{saksyuconvex}, i.e., domains where social choice function that satisfy weak monotonicity are necessarily truthful.\footnote{A function $f$ satisfies weak monotonicity if for every player $i$, $v_{-i}\in V_{-i}$ and $v,v'\in V_i$, if $f(v,v_{-i})=a$ and $f(v',v_{-i})=b$, it implies that $v(a)-v(b)\ge v'(a)-v'(b)$.} We will prove properties of convex domains by reducing them to scalable ones: we show that for every function, translating its domain by a constant has no effect on it, and that all convex domains translate to scalable domains.
\begin{definition} (Translation)
	Let $f\colon V= V_1\times \cdots\times V_n \to \mathcal{A}$ and $f^t\colon V^t= V_1^t\times \cdots\times V_n^t \to \mathcal{A}$ be two social choice functions. Then, $(f,V)$ and $(f^t,V^t)$ are \emph{translations} if there exist $n$ vectors $t_1,\ldots,t_n\in \mathbb{R}^{|\mathcal{A}|}$ such that:
	\begin{enumerate}
		\item For every player $i$,  $V_i^t=\{v_i-t_i|v_i\in V_i\}$.
		\item For every $(v_1^t,\ldots,v_n^t)$, $f^t(v_1^t,\ldots,v_n^t)=f(v_1^t+t_1,\ldots,v_n^t+t_n)$.  
	\end{enumerate} 
We write $t_i^a$ for the coordinate of alternative $a$ in the translation vector of player $i$. 
\end{definition}
\begin{lemma} \label{translationpreservespaymentslemma}
	Let $(f,V)$ and $(f^t,V^t)$ be translations of one another with the vectors $t_1,\ldots,t_n$. Then, if the payment scheme $P$ implements (deterministically or in expectation) $f$, the following payment implements $f^t$ (deterministically or in expectation) :
	\begin{equation} \label{thirteenandhalf}
	P_i^t(v_1^t,\ldots,v_n^t)=P_i(v_1^t+t_1,\ldots,v_n^t+t_n)+t_i^a
	\end{equation}
	where  $a=f^t(v_1^t,\ldots,v_n^t)$. 
\end{lemma}
\begin{proof}
	We prove for deterministic implementation, and the proof for for truthful in expectation implementation is identical. Fix a player $i$ and $v_{-i}^t \in V_{-i}^t$. We want to show that for all $v^t,\hat{v}^t \in V_i^t$:
	\begin{equation} \label{goal}
	v^t(f^t(v^t,v_{-i}^t))-P_i^t(v^t,v_{-i}^t) \ge v^t(f^t(\hat{v}^t,v_{-i}^t))-P_i^t(\hat{v}^t,v_{-i}^t)   
	\end{equation}
	We denote $v=v^t+t_i$ and $\hat{v}=\hat{v}^t+t_i$. Similarly, $v_{-i}=(v^t_1 +t_1,\ldots,v^t_{i-1} +t_{i-1},v^t_{i+1} +t_{i+1},\ldots,v^t_{n} +t_{n})$. By definition, $f^t(v^t,v_{-i}^t)=f(v,v_{-i})$ and  $f^t(\hat{v}^t,v_{-i}^t)=f(\hat{v},v_{-i})$. We denote these alternatives with $a$ and $\hat{a}$, respectively. Similarly, by (\ref{thirteenandhalf}), $P^t_i(v^t,v_{-i}^t)=P_i(v,v_{-i})+t_i^a$ and  $P^t_i(\hat{v},v_{-i}^t)=P_i(\hat{v},v_{-i})+t_i^{\hat{a}}$. Hence, (\ref{goal}) is equivalent to:
	\begin{equation} \label{instrument}
	v^t(\underbrace{f(v,v_{-i})}_{a})-P_i(v,v_{-i}) - t_i^a \ge v^t(\underbrace{f(\hat{v},v_{-i})}_{\hat{a}})-P_i(\hat{v},v_{-i}) - t_i^{\hat{a}}  
	\end{equation}
	By definition, $v^t(a)=v(a)+t_i^a$ and $v^t(\hat{a})=v(\hat{a})+t_i^{\hat{a}}$. Hence, the truthfulness of $(f,P)$ implies that (\ref{instrument}) holds. (\ref{instrument}) holds only if (\ref{goal}) holds, so $P^t$ implements $f^t$. 
\end{proof}
We show that $f$ and $f^t$ require the same amount of communication for truthful implementation. 
\begin{lemma} \label{translationsamecommcomplexitylemma}
	Let $(f,V)$ and $(f^t,V^t)$ be translations of one another. Then, $cc(f)=cc(f^t)$, $cc_{TIE}(f)=cc_{TIE}(f^t)$ and $cc_{IC}(f)=cc_{IC}(f^t)$.
\end{lemma}
\begin{proof}
	It is clear that $cc(f)=cc(f^t)$ for every pair of translations $f$ and $f^t$. 
	We explain why $cc_{IC}(f)=cc_{IC}(f^t)$ and the proof of  $cc_{TIE}(f)=cc_{TIE}(f^t)$ is identical.  
	The proof is by reduction of  $(f^t,V^t)$ to $(f,V)$. Denote their translations with $t_1,\ldots,t_n$, and the protocol of a truthful mechanism for $f$ with $\pi$. For 
	$(v_1^t,\ldots,v_n^t)\in V^t$, the players run $\pi(v_1^t+t_1,\ldots,v_n^t+t_n)$. 
	By construction,  $f(v_1^t+t_1,\ldots,v_n^t+t_n)=f^t(v_1^t,\ldots,v_n^t)=a$. For the payment of every player $i$, the players output $P_i+t_i^a$, where $a=f(v_1^t+t_1,\ldots,v_n^t+t_n)$. It  
 is truthful by Lemma \ref{translationpreservespaymentslemma}.     
\end{proof}
\begin{lemma} \label{convextranslatesscalablelemma}
	Let $f\colon V=V_1\times ..\times V_n$ social choice with convex domains. Then, $(f,V)$ has a translation $(f^t,V^t)$ such that $V^t_1,\ldots,V^t_n$ are scalable.  
\end{lemma}
\begin{proof}
	The translation is as follows. Fix a function $f$ and a convex domain $V$.
	For every domain $V_i$, we take an arbitrary type $v_i^\ast \in V_i$ and form a translation of $V_i$ by taking $V_i^t=\{v_i- v_i^\ast | v_i \in V_i \}$ and setting $f^t(v_1^t,\ldots,v_n^t)=f(v_1^t+v_1^\ast,\ldots,v_n^t+v_n^\ast)$ for all $(v_1^t,\ldots,v_n^t)\in V_1^t \times \cdots \times V_n^t$. We want to show that for all $i$,   $V^t_i$ is scalable.  $V^t_i$ is convex, because translating sets preserves convexity. Also, by construction it contains $\vec{0}$. Hence, $V^t_i$ is scalable because every $v\in V^t_i$ and every $\lambda \in [0,1]$ satisfy by convexity  that $\lambda v + (1-\lambda)\cdot \vec{0} = \lambda v \in V_i^t$.    
\end{proof}
\begin{theorem}\label{convexalgorithmthm}
	Let $f\colon V= V_1\times ..\times V_n \to \mathcal{A}$ be a social choice function with convex domains. Then, $cc_{TIE}(f)\le (n+1)\cdot cc(f)$. 
\end{theorem}
\begin{proof}
	By Lemma \ref{convextranslatesscalablelemma}, $(f,V)$  has a translation $(f^t,V^t)$ such that $V^t$ is scalable. Therefore:
	\begin{equation*}
	cc_{TIE}(f)\underbrace{=}_\text{By Lemma \ref{translationsamecommcomplexitylemma}} cc_{TIE}(f^t) \underbrace{\le}_\text{by Theorem \ref{scalableinexpectationthm}} (n+1)\cdot cc(f^t) \underbrace{=}_\text{by Lemma \ref{translationsamecommcomplexitylemma}} (n+1)\cdot cc(f)
	\end{equation*}
\end{proof}

\section{An Algorithm for Single Parameter Settings}\label{gsalgorithmsection}
We now return to considering \emph{deterministic} ex-post implementations. 
%We prove the following: let $f\colon V= V_1\times \cdots \times V_n \to \mathcal{A}$ where $V$ is a single parameter domain. Then, $cc_{IC}(f)\le poly(n,|\mathcal{A}|,cc(f))$. As a corollary, in binary single parameter domains we get that $cc_{IC}(f)= \mathcal{O}(n\cdot cc^{2}(f))$. In particular, this upper bound holds for all the functions in the setting considered in \cite{schapira}. 
%See Section \ref{gsalgorithmsection}.
Notice that the exponential lower bounds of Section \ref{lowerboundsec} were proven using single parameter social choice functions. We provide an algorithm for all such domains. The upper bound on the communication complexity of the algorithm  has a linear in $|\mathcal{A}|$ factor.
The communication complexity of the algorithm is optimal in the sense that the dependence on $|\mathcal{A}|$ is necessary, as demonstrated by the examples in Theorems \ref{mostimportanttheorem} and \ref{nopolylogthm}.  

\begin{theorem} \label{singleparamalgthm1}
	For all single parameter environments, $cc_{IC}(f)= \mathcal{O}(n\cdot cc^{2}(f)\cdot |\mathcal{A}|)$. As a corollary, for binary single parameter domains, $cc_{IC}(f)= \mathcal{O}(n\cdot cc^{2}(f))$.
\end{theorem}

\begin{proof}
	Recall that in single parameter settings, the valuations set of each player is composed of a public function $w_i\colon\mathcal{A}\to \mathbb{R}$ and a type space which contains scalar private information $R_i$. Each valuation $v_i(\cdot)$ is equal to $r_i\cdot w_i(\cdot)$ for some $r_i\in R_i$. For brevity, throughout the proof we slightly abuse notation by writing $v_i$ both for a valuation and for the scalar private information associated with it, $r_{v_{i}}$.

	We will show that $cc_{IC}(f)\le \mathcal{O}(n\cdot cc^2(f)\cdot \max_i|\Ima w_i|)$. It implies the theorem, because $\max_i|\Ima w_i| \le |\mathcal{A}|$. Since binary single parameter functions satisfy that $\Ima w_{i}=\{0,1\}$ for every player $i$, it is immediate that they satisfy $cc_{IC}(f)= \mathcal{O}(n\cdot cc^{2}(f))$. 
	  
The proof is based on the observations that the payment fully depends on $w_i(f(\cdot,v_{-i}))$ and that by multiple binary searches, the players know $w_i(f(v_i,v_{-i}))$ for all $v_i \in V_i$. The binary searches are for the sake of finding the \textquote{threshold} values of each alternative.  
%The proof is based on the observations that the payment fully depends on $w_{i}(f(\cdot,v_{-i}))$, and that computing it can be done by multiple binary searches: one for the \textquote{threshold} of each alternative. 
The theorem has no assumptions at all about neither the domain nor the function, but it comes at a price: the proof involves a lot of technicalities in order to include all single parameter domains. 
\begin{lemma}\label{wsufficeslemma}
	Let  $f\colon V_1\times \cdots\times V_n\to \mathcal{A}$ be an implementable social choice function. Then, there exists a payment scheme $P$ which implements $f$ such that for every player $i$ and and every $v_{-i}^{1},v_{-i}^{2}\in V_{-i}$:
	\begin{equation}\label{cond}
	w_{i}(f(\cdot,v_{-i}^{1}))\equiv w_{i}(f(\cdot,v_{-i}^{2})) \implies P_{i}(\cdot,v_{-i}^{1})\equiv P_{i}(\cdot,v_{-i}^{2}) 
	\end{equation}   
\end{lemma}
	If two types in $V_{-i}$ have the same effect on the alternative chosen for player $i$, then they can clearly be implemented with the same payment. Therefore,  $cc_{IC}(f)$
	is at most the communication complexity of fully characterizing the function $w_{i}(f(\cdot,v_{-i}))$ for every player $i$.    

Fix a player $i$. Denote the number of elements in $\Ima w_{i}$ with $m$, and the elements in $\Ima w_{i}$ with $w_{1},\ldots,w_{m}$. For brevity, from now on we call $w_{i}(f(\cdot,v_{-i}))$ simply $w(\cdot,v_{-i})$. Note that each $v_{-i}\in V_{-i}$ defines a partition of $V_{i}$ to $\mathcal{V}_{1}\cupdot\ldots\cupdot \mathcal{V}_{m}$ where $\mathcal{V}_{j}=\{v_{i}\big|w_i(v_{i},v_{-i})=w_{j}\}$. We define the infimum of unbounded from below sets as $-\infty$, and the supremum of unbounded from above sets as $\infty$, so the infimum and supremum of $\mathcal{V}_{j}$ are always well defined. 
	 Clearly, knowing the partition of $V_{i}$ that $v_{-i}$ induces $\mathcal{V}_{1}\cupdot\ldots\cupdot \mathcal{V}_{m}$ is equivalent to fully computing $w(\cdot,v_{-i})$. $f$ is monotone, so we can focus on computing the \textquote{thresholds} of those subsets:
	\begin{lemma}\label{howtolearnwlemma}
		Let $v_{-i}^{1},v_{-i}^{2}\in V_{-i}$ be two types, and denote the partition of $V_{i}$ that they induce as $\mathcal{V}_{1}^{1}\cupdot\ldots\cupdot \mathcal{V}_{m}^{1}$ and as $\mathcal{V}_{1}^{2}\cupdot\ldots\cupdot \mathcal{V}_{m}^{2}$, respectively. Suppose that for all $j\in[m]$, 
			 $\inf \mathcal{V}_{j}^{1}=\inf \mathcal{V}_{j}^{2}$ and 
  $w(\inf \mathcal{V}_{j}^1,v_{-i}^{1})=w(\inf \mathcal{V}_{j}^2,v_{-i}^{2})$ (if  
  $\inf \mathcal{V}_{j}^1, \inf \mathcal{V}_{j}^2$ belong in $V_{i}$).
		Then, $w(\cdot,v_{-i}^{1})\equiv w(\cdot,v_{-i}^{2})$.  
	\end{lemma}
	For every $j\in[m]$, we define a function $inf_{j}:V_{-i}\to \mathbb{R}\cup\{\pm\infty\}$ which returns $\inf \mathcal{V}_{j}^{v_{-i}}$. We  define  $sup_{j}:V_{-i}\to \mathbb{R}\cup\{\pm\infty\}$ similarly. 
	\begin{lemma} \label{alggspfoolingset}
		For every $j\in [m]$,	$cc(f)\ge \log |\Ima inf_{j}|$ and $cc(f)\ge \log |\Ima sup _{j}|$. Hence, $\Ima inf_{j}$ and $\Ima sup_{j}$ are finite.
	\end{lemma}
	\begin{proof}[Proof of Lemma \ref{alggspfoolingset}]
		The first part of the lemma implies its second part, because $cc(f)$ is assumed to be finite (Remark \ref{finiteccfremark}). 
		The proof is by a projection of the single-parameter domain to a multi-dimensional domain. 
		
		We \textquote{spread} each $v\in V_i$ to an $m-$dimensional representation: $(v(1),\ldots,v(m))$ where $v(j)=r_v\cdot w_j$. We unite the alternatives by their $w_i(\cdot)$ values, so there are now $m$ alternatives, denoted with $w_1,\ldots,w_m$. Thus, the value of player $i$ for alternative $w_j$ is $v(j)$.  We define $\delta_{j1}=\inf \{v(j) - v(1) \hspace{0.25em}|\hspace{0.25em}f(v,v_{-i})=w_j\}$. Notice that if $\inf_j(v_{-i})=x$ in the single dimensional perspective, it means that $\delta_{j1}(v_{-i})=x(w_j-w_1)$. Therefore,  $|\Ima inf_j| = |\Ima \delta_{j1}|$. By Lemma \ref{Rabccfbound}, $cc(f)  \ge \log |\Ima \delta_{j1}|$, so $cc(f)\ge \log |\Ima inf_j|$. $cc(f)\ge \log \Ima sup_j$ is proved  analogously.  	
	\end{proof}
	\begin{lemma}\label{protocollemma}
		For every $j\in[m],$ $cc(inf_j)\le cc^{2}(f)$. 
	\end{lemma} 
	\begin{proof}
		We describe a protocol which is  a variant of binary search over the elements in $\Ima inf_j$. Denote the elements in $\Ima inf_j$ which differ from $\infty$ as $i_{1}<\ldots<i_{r}$.\footnote{$\Ima inf_j$ is countable by
			Lemma \ref{alggspfoolingset}.} Note that the infimums do not necessarily belong in $V_i$, so we cannot just do a binary search over the set of infimums. In order to overcome this problem, we need to find types in $V_{i}$ such that the following conditions hold:
		\begin{enumerate}
			\item $i_{1}\le v_1<i_2\le\ldots<i_{r}\le v_r$.
			\item For every $v_{-i}\in V_{-i}$ such that $inf_j(v_{-i})=i_l$, $w(v_{l},v_{-i})=w_j$ or
			$w(i_l)=w_j$.  
		\end{enumerate}
		Fix an index $1 \le l \le r$.
		For $i_{l+1}$ to be always well defined, we set $i_{r+1}$ as $\infty.$ Define:
		\begin{equation} \label{minset}
		sup(l)= \min\{ sup \in \Ima sup_j \hspace{0.3em}|\hspace{0.3em} \exists v_{-i} \hspace{0.3em} \text{s.t.} \hspace{0.3em} inf_j(v_{-i})=i_l \hspace{0.25em} \land sup_j(v_{-i})=sup \hspace{0.25em} \land \hspace{0.25em} sup\neq i_l  \}
		\end{equation}
		In words, this is the smallest supremum of $\mathcal{V}_j$ whenever $inf_{j}(v_{-i})=i_l$, which is strictly larger than $i_l$.\footnote{ $\Ima sup_j$ is finite by Lemma \ref{alggspfoolingset}, so the set in (\ref{minset}) has a minimum if it is not empty.} If the subset is empty, $sup(l)\gets \infty$.  Denote $\min\{sup(l),i_{l+1}\}$ with $min(l)$. Both $sup(l)$ and $i_{l+1}$ are by definition strictly larger than $i_l$, so $i_l<min(l)$. 
		If there exists $v\in V_{i}$ such that $i_l<v<min(l)$, we take it to be $v_{l}$. Otherwise, $v_{l}\gets i_l$.  
		
		 We explain why the chosen $v_1,\ldots,v_r$ satisfy the desired properties. First, it is immediate from the construction that $i_l\le v_{l}<i_{l+1}$ for all $l\in[r]$. For the second condition, let $v_{-i}\in V_{-i}$ be a type such that $inf_{j}(v_{-i})=i_l$. Recall that $v_{-i}$ defines a partition of $V_{i}$ to $\mathcal{V}_{1}\cupdot\ldots\cupdot \mathcal{V}_{m}$. By definition, $\inf \mathcal{V}_{j}=i_l$. If $inf_{j}(v_{-i})=sup_{j}(v_{-i})=i_l$, it implies that $\mathcal{V}_{j}$ is a singleton, so $\mathcal{V}_{j}=\{i_{l}\}$ and $w(i_{l},v_{-i})=w_j$, as needed. Otherwise,$v_{-i}$ satisfies that $inf_j(v_{-i})< sup_j(v_{-i})$. 
		 We now handle two cases separately:  $i_l<v_l$ and $i_l=v_l$.
		 If $i_l<v_l$: 
		 \begin{equation*}
		 inf_j(v_{-i})=i_l<v_l \underbrace{<}_\text{by construction} min(l)\le sup(l) \le sup_j(v_{-i}) 
		 \end{equation*}
		 which implies that $w(v_l,v_{-i})=w_j$, as needed.
		If $i_l=v_l$, the way we have chosen $v_{l}$ allows us to deduce that there is no type in $V_{i}$ that belongs in the interval $(i_l,min(l))$. Since $\mathcal{V}_{j}\subseteq V_{i}$, it implies that $\mathcal{V}_{j}\cap (i_l,min(l))=\varnothing$. However, $i_l$ is the greatest lower bound of $\mathcal{V}_{j}$, so clearly $i_l\in \mathcal{V}_{j}$ and thus $w(i_l,v_{-i})=w_j$.  
		
		Then, using these types,  we do a binary search for $inf_j(v_{-i})$, by looking up for an element $i_l$ such that $w(v_{l-1},v_{-i})<w_{j}$ and $w(v_{l},v_{-i})=w_{j}$ or $w(i_l,v_{-i})=w_j$. Due to monotonicity, $w(v_{l-1},v_{-i})<w_{j}$ implies that $inf_j(v_{-i})> i_{l-1}$. If $w(v_{l},v_{-i})=w_{j}$ or $w(i_l)=w_j$, it implies that $\mathcal{V}_{j}$ contains $i_l$ or $v_l$, so $inf_{j}(v_{-i})<i_{l+1}$. Hence, $inf_{j}(v_{-i})=i_{l}$. Clearly, if $inf_{j}(v_{-i})=\infty$, there is no index $l$ that satisfies those conditions and the protocol outputs $\infty$. It is easy to see that the properties of the sequence $i_1\le v_1<\ldots<i_r\le v_r$ guarantee that whenever the algorithm focuses on the left half or on the right half of the sequence $i_1,\ldots,i_r$, the element that we look for,  $inf_{j}(v_{-i})$, belongs in it (if it differs from $\infty$). 
		
		As for communication,
		each step in the binary search requires at most $2\cdot cc(f)$ bits due to the computations of  $w(f(v_{l},v_{-i}))$ and $w(f(i_{l},v_{-i}))$. Combining well known properties of binary search with Lemma \ref{alggspfoolingset} yields that there are at most $\mathcal{O}(\log r)= \mathcal{O}(\log |\Ima inf_{j}|)=\mathcal{O}(cc(f))$ steps. Hence, the total communication of computing $inf_j$ is $\mathcal{O}(cc^{2}(f))$ bits.
	\end{proof}	
Thus, computing $inf_j(\cdot)$ and executing $f$ on the instance $(inf_j(v_{-i}),v_{-i})$ for all $j\in [m]$  takes at most $\mathcal{O}(m\cdot cc^2{f})$ bits. By Lemmas \ref{wsufficeslemma} and \ref{howtolearnwlemma}, it suffices for payment computation. Recall that $m=|\Ima w_{i}|$. By repeating for all players, we get that the total communication of the suggested protocol is $\mathcal{O}(n\cdot cc^{2}(f)\cdot \max_{i} |\Ima w_{i}|)$. 
		\end{proof}
\paragraph{Tightness.} We explain why the factors $|\mathcal{A}|$ and $cc(f)$ cannot be omitted, i.e. that it cannot be the case that for all functions, or even for all single parameter functions that $cc_{IC}(f)\le poly(n,cc(f))$ or $cc_{IC}(f)\le poly(n,|\mathcal{A}|)$. For non-degenerate functions $cc(f)\ge n$, so we consider the $n$ factor to be less significant.  The social choice functions in Section \ref{lowerboundsec}  serve as counterexamples to $cc_{IC}(f)\le poly(n,cc(f))$. All of them satisfy that $cc_{IC}(f)=\exp(cc(f))$ with a constant number of players. 
% The first satisfies $n=2$ and $cc(f)=\mathcal{O}(k)$, but $cc_{IC}(f)=\exp(k)$, whereas for the latter $cc(f)=\mathcal{O}(\log k)$
%and $n=3$, but $cc_{IC}(f)=\Omega(k)$.  

 Similarly, we can easily provide a function with two alternatives and two players where the communication complexity of its implementation is arbitrarily large: let
	$f_k\colon  V_1\times V_2\to \{a_0,a_1\}$. The valuations of the players are their private information, and they do not depend on the alternative chosen:  $R_{1}=V_{1}=R_2=V_{2}=\{0,1,\ldots,2^{k}-1\}$. $f_k(v_1,v_2)=a_1$ if and only if the bit representations of $v_{1},v_{2}$ are disjoint. Clearly, $f_k$ is harder than the function $DISJ_{k}$ and it is well known that $cc(DISJ_{k})=\Omega(k)$ \cite{ccbook}. $f_k$ is implementable with no payments, because the valuations of  both players do not depend on the outcome. Hence, $cc_{IC}(f_k)= cc(f_k) = \Omega(k)$, whereas $n=2$ and $|\mathcal{A}|=2$.
%\begin{openquestion}
%	Is there a function $f$ such that $cc_{IC}(f)\le poly(n,max_{i} |\Ima w_{i}|,cc(f))$, but $cc_{IC}(f)= \omega (max_{i} |\Ima w_{i}|,cc(f))$? In other words, is the factor of $n$ necessary? 
%\end{openquestion}

\section{Payment Computation in Multi-Parameter Settings}\label{uniquesection}
So far, we considered deterministic algorithms only for single-parameter domains. In this section, we venture into the more challenging multi-parameter setting. We begin by proving an efficient algorithm for functions that satisfy uniqueness of payments.   
There is a vast literature on the topic of characterizing domains and functions where implementability guarantees uniqueness of payments (for example, \cite{vohrarevneueequiv,noamchapter}). Notice that \textquote{uniqueness of payments} is often called \textquote{revenue equivalence}. 

We conclude by showing that an efficient algorithm for functions that satisfy uniqueness of payments yields efficient algorithms for implementable functions with scalable and convex domains (Claims \ref{scalabledeterministicthm} and \ref{convexdetclaim}).
%We also show that the upper bound we provide for uniqueness of payments can be extended to convex domains.  
%\paragraph{Multi-Parameter Domains.} We now consider the challenging case of multi-parameter domains. We consider domains that satisfy the uniqueness of payments property. In these domains, for every two implementations of a social choice function $f$ by $P_1,\ldots, P_n$ and by $P'_1,\ldots, P'_n$, it holds that for every $i$ and $v_{-i}: P_i(\cdot, v_{-i})-P_i(\cdot, v_{-i})=c$, for some constant $c$. 
%We also observe that all scalable and convex domains satisfy uniqueness of payments.

%For all those domains we prove the following: Let $f\colon V_{1}\times\cdots\times V_n \to\mathcal{A}$ be an implementable social choice function which satisfies uniqueness of payments. Then, $cc_{IC}(f)\le poly(n,|\mathcal{A}|,cc(f))$.

The intuition to the proof is as follows. Instead of providing a deterministic protocol that proves this bound as usual, we provide a non-deterministic protocol that shows it. A deterministic mechanism follows by relying on the known fact that the connection between deterministic and non-deterministic mechanisms is polynomial (in fact, this known fact was not proven for promise problems that are needed in our proofs, so we extend the connection to hold also for promise problems in Section \ref{promisesection}).

Thus, the problem boils down to providing a succinct witness that determines the price of the altrnative chosen. We make the observation that a payment can be determined by a conjunction of $\mathcal{O}(|\mathcal A|^2)$ inequalities. For illustration, fix $v_{-i}$. Let $v^a$ be such that $f(v^a,v_{-i})=a$ and $v^b$  such that $f(v^b,v_{-i})=b$. Then, it obviously holds that $v^a(a)-p^a \geq v^a(b)-p^b$ and $v^b(b)-p^b \geq v^b(a)-p^a$, so $v^a(a)-v^a(b)\geq p^a-p^b \geq v^b(a)- v^b(b)$. We show that if we choose, for each such $a$ and $b$, $v^a$ and $v^b$ such that the inequality is ``tight'' as possible, then the payment of an alternative can be determined. These $\mathcal{O}(|\mathcal A|^2)$ types will serve as our non-deterministic witness, which completes the proof except that the description of the types might be huge. We rely on the communication protocol of $f$ to provide a succinct description of them. 
%See Section \ref{uniquesection} for this result.
We begin with some formalities. 
\begin{definition} 
	(Uniqueness of Payments)
	A social choice function $f\colon V=V_{1}\times\cdots\times V_n \to\mathcal{A}$ satisfies uniqueness of payments if for every pair of truthful mechanisms $(f,P)$ and $(f,P')$, it holds that there exist $n$ functions $h_{1},..,h_{n}$ where $h_{i}:V_{-i}\to \mathbb{R}$, such that for every player $i$ and every $(v_{1}...,v_{n})\in V$:
	\begin{equation}\label{paymentsuniqueeq}
	P_{i}(v_{1},\ldots,v_{n})=
	P'_{i}(v_{1},\ldots,v_{n})+h_{i}(v_{-i})
	\end{equation}
\end{definition}
In Section \ref{setupmyersonsection}, we define normalized mechanisms for single parameter settings. We generalize the definition to multi-parameter domains.  
\begin{definition} \label{normalizempdef} (Multi-Parameter Normalization)
For every player $i$, let $v_i^0 \in V_i$ be its zero type. We say that a mechanism $M=(f,P)$ is \emph{normalized} if for every player $i$ and every $v_{-i}\in V_{-i}$, $P_i(v_i^0,v_{-i})=0$. 
\end{definition}
For every $v_{-i}$, define $f(v_0,v_{-i})=a_0$ as the \emph{zero alternative} with respect to $v_{-i}$. 
It is easy to see that if $f$ satisfies uniqueness of payments, there exists a unique normalized mechanism which implements it. 
%We begin by using a reduction of deterministic payment computation to nondeterministic payment computation and stating the central claim, Claim \ref{claimstep2}. 
\begin{theorem} \label{uniquepaymentsalgthm}
Let $f\colon V_1\times \cdots \times V_n \to \mathcal{A}$ be an implementable social choice function that satisfies uniqueness of payments. Then, $cc_{IC}(f)\le  poly(n,|\mathcal{A}|,cc(f))$. 
\end{theorem}
\begin{proof}[Proof of Theorem \ref{uniquepaymentsalgthm}]

	Fix an implementable function $f$ that satisfies uniqueness of payments, and denote the normalized mechanism for it as $M$. We show an upper bound for $cc_{IC}(f)$ by presenting a communication protocol for $M$. Computing the outcome requires $cc(f)$ bits, so clearly the tricky part is computing the payments. 
	By the taxation principle, the payment of every player is a function of $v_{-i}$ and of the alternative chosen. Hence,  we define a promise function for the price of an alternative $a\in \mathcal{A}$ given $v_{-i}$, $price_{i}^{a}:V_{-i}\to \mathbb{R}$, with the promise that $a$ is reachable from $v_{-i}$.\footnote{All elements in  $\Ima price^{i}_{a}$ are real. The reason for it is as follows. The price of a reachable alternative $a$ given $v_{-i}$ cannot be $-\infty$, because then a player would be incentivized to misreport his type whenever it is $v_{0}$, by reporting instead a type that reaches $a$. Similarly, the price  cannot be $\infty$, because in this case the player would deviate from truthfulness whenever his real type is the type that reaches $a$.}   
	\begin{claim}\label{claimstep2}
		For every player $i$, $a\in \mathcal{A}$ and price $p\in \Ima price_{i}^{a}$, it holds that
		$N(price_i^{a}) \le \mathcal{O}(|\mathcal{A}|^{2} \cdot cc(f))$.
	\end{claim}
	The proof of Claim \ref{claimstep2} can be found in Section \ref{claimstep2subsec}. 
	Combining the polynomial relation between deterministic and non-deterministic communication complexity (Section \ref{promisesection}) with Claim \ref{claimstep2}, we get that: 
%	Claim \ref{claimstep2} implies Theorem \ref{uniquepaymentsalgthm}:
	\begin{equation*}
	cc_{IC}(f) \le \underbrace{cc(f)}
	_{\substack{\text{computing}\\ {f(v=a)}}} 
	+ \sum_{i=1}^{n} cc(price_{i}^{a})  
	\underbrace{\le}_{\substack{\text{by}\\\text{Theorem \ref{multipartynondtodreductionpromise}}}}
	cc(f) + \sum_{i=1}^{n} poly(n,N(price_i^a)) 
	\underbrace{\le}_{\substack{\text{by}\\\text{Claim \ref{claimstep2}}}}
	poly(n,|\mathcal{A}|,cc(f)) 
	\end{equation*} 
	Observe that $f(v)=a$ implies that for every player $i$, $a$ is reachable from $v_{-i}$, so $v_{-i}$ satisfies the promise. 
%	By Definition \ref{promisedefinition},
%	 $N(price_i^a)$ is the communication complexity of presenting a witness for every price presented by $v_{-i}$ for alternative $a$ if $a$ is reachable from $v_{-i}$, such that if $a$ is not reachable  all witnesses  are rejected. 
\end{proof}
\subsection{Proof of Main Claim} \label{claimstep2subsec}
\begin{proof}[Proof of Claim \ref{claimstep2}]

Fix a player $i$ and an alternative $a^{\ast}$. We will prove the claim by presenting a proof system for $price_{i}^{a^{\ast}}$. Since $price_{i}^{a^{\ast}}$ is a promise function, a valid proof system for it satisfies that if $v_{-i}$ breaks the promise, i.e., $a^{\ast}$ is not reachable from $v_{-i}$, then  the players reject all witnesses for it.\footnote{See Definition \ref{promisedefinition}.}  

First, we make some definitions and prove useful properties. Following the notation of \cite{bikhchandani}, let $\delta_{ab}:V_{-i}\to\mathbb{R}\cup\{\pm\infty\}$ be a function  that maps $v_{-i}\in V_{-i}$ to  $\inf\{v(a)-v(b)\hspace{0.25em}|\hspace{0.25em}f(v,v_{-i})=a\}$ for every pair of alternatives $a,b\in \mathcal{A}$.\footnote{If the  set defined by $v_{-i}$ is not bounded from below,  $\delta_{ab}(v_{-i})$ outputs $-\infty$.} Because of truthfulness, $\inf\{v(a)-v(b)\hspace{0.25em}|\hspace{0.25em}f(v,v_{-i})=a\}$ is un upper bound of the difference between the payment of $a$ and the payment  of $b$ in every payment scheme that implements $f$. 
We begin with the following technical lemma:
\begin{lemma} \label{Rabccfbound}
	Let $f\colon V_1\times \cdots\times V_n\to\mathcal{A}$ be a social choice function.Then, for every player $i$ and every pair of alternatives $a,b\in \mathcal{A}$, $cc(f)\ge \log |\Ima \delta_{ab}|$. In particular, $\Ima \delta_{ab}$ is finite. 
\end{lemma}
\begin{proof}[Proof of Lemma \ref{Rabccfbound}] First, $cc(f)\ge \log|\Ima \delta_{ab}|$ implies that $\Ima \delta_{ab}$ is finite, because we assume that $cc(f)$ is finite (see Remark \ref{finiteccfremark}). In order to prove $cc(f)\ge \log |\Ima \delta_{ab}|$, we take an arbitrary \emph{finite} subset $T$ of $\Ima\delta_{ab}$. Denote the $t$ elements with $\delta_{1}<\ldots<\delta_{t}$. The proof is by the fooling set method. 
	For every ${j}\in[t]$, we take a type $v_{-i}^{j}$ such that $\delta_{ab}(v_{-i}^{j})=\delta_{j}$. We pair it with a type $v^{j}_{i}\in V_i$ in the following way. If $\delta_t=\infty$, we take $v_i^{t}$ to be an arbitrary type in $V_i$. Otherwise, we take $v_i^{t}$ such that $f(v^{t}_{i},v_{-i}^{t})=a$ and $\delta_t \le v^{t}_i(a)-v^{t}_i(b)$. For $j<t$, we take $v^{j}_i$ such that $f(v^{j}_{i},v_{-i}^{j})=a$
	and $\delta_{j}\le v^{j}_{i}(a)-v^{j}_{i}(b)<\delta_{j+1}$. 
	
	We need to show that every $v^{j}_{-i}$ has a matching  $v^{j}_{i}$ that satisfies those requirements. 
	If $\delta_t=\infty$, it is trivial that an arbitrary type in $V_i$ exists. If $\delta_t<\infty$, a matching $v_i^t$ for $v_{-i}^t$ necessarily exists because $\{v(a)-v(b)\hspace{0.25em}|\hspace{0.25em}f(v,v_{-i}^{t})=a\}$ is non empty and $\delta_{t}$ is its real lower bound. For $j<t$, such a type necessarily exists, since $\delta_{j}$ is by definition the greatest lower bound of $\{v(a)-v(b)\hspace{0.25em}|\hspace{0.25em}f(v,v_{-i}^{j})=a\}$. The fact that $\delta_{j+1}$ is not the infimum even though it is larger than $\delta_{j}$ implies that it is not a lower bound of the subset, so there exists a type $v_{i}^{j}$ such that $\delta_{j} \le v_{i}^{j}(a)-v_{i}^{j}(b)<\delta_{j+1}$ and $f(v_{i}^{j},v_{-i}^{j})=a$.
	
	Let $S = \{(v^{j}_{i},v_{-i}^{j}) \big| 1 \le j \le t\}$. We will show that any two inputs in $S$ cannot belong in the same leaf in any communication protocol for $f$. Let $(v_i^k,v_{-i}^k),(v_i^l,v_{-i}^l)$ be two inputs in $S$ where $k<l$. If $\delta_{l}=\infty$, clearly $f(v_i^l,v_{-i}^l)\neq a$, whereas $f(v_i^k,v_{-i}^k)=a$. Then, they clearly do not belong in the same leaf. If $\delta_{ab}(v_{-i}^l)<\infty$, we know that $f(v_i^k,v_{-i}^k)=(v_i^l,v_{-i}^l)=a$. 
	Note that
	$\delta_{k} \le v^{k}_{i}(a)-v^{k}_{i}(b)< \delta_{k+1}$.\footnote{$\delta_{k+1}$ is well defined because $k<l$.} $k<l$ means that $\delta_{k+1}\le \delta_{l}$. Combining these two inequalities yields:
	$	v^{k}_{i}(a)-v^{k}_{i}(b)<\delta_{l}=\inf\{v(a)-v(b)\hspace{0.25em}|\hspace{0.25em}f(v,v_{-i}^{l})=a\}$. 
	Thus, by the definition of infimum,  $f(v_{i}^{k},v_{-i}^{l})\neq a$ so  $(v^{k}_{i},v_{-i}^{k})$ and $(v^{l}_{i},v_{-i}^{l})$ cannot belong in the same leaf because they violate the mixing property. Hence, the number of leaves is at least $t$, so $cc(f)\ge \log t$. If every finite subset $T$ of the set $\Ima\delta_{ab}$ is of size at most $2^{cc(f)}$, then the size of the set itself is at most $2^{cc(f)}$. Therefore, $2^{cc(f)}\ge |\Ima\delta_{ab}| \implies cc(f)\ge \log |\Ima\delta_{ab}|$. 
\end{proof}	
\begin{proposition} \label{asquaredtypeslemma}
	Let $f$ be function that satisfies uniqueness of payments and suppose that $\Ima \delta_{ab}$ is countable for all $a,b\in \mathcal{A}$. Then, for every 
	$v_{-i}$ and for every reachable alternative from it $a^\ast$, there exists a set $R\subseteq V_i$ of types $v_1,\ldots,v_t$, $t=\mathcal{O}(|\mathcal{A}|^2)$, such that if $v_{-i}'\in V_{-i}$ satisfies $f(v,v_{-i})=f(v,v_{-i}')$ for all $v\in R$, then $price_i^{a^\ast}(v_{-i})=price_i^{a^\ast}(v_{-i}')$. Moreover, $R$ contains the zero type $v_0$ and a type $v^{\ast}$ such that $f(v^{\ast},v_{-i})=a^\ast$.     
\end{proposition}
Fix a function $f$ that satisfies those assumptions.  By the taxation principle, we interpret a truthful mechanism as a process where a menu $\mathcal{M}_{v_{-i}}$ is presented to every player $i$ with 
her prices for all alternatives,
and then an alternative which maximizes  the utilities of  all players is chosen. We assume that if an alternative is not reachable from $v_{-i}$, its price in the menu is $\infty$ and otherwise by definition $\mathcal{M}_{v_{-i}}(a)=price_i^a(v_{-i})$.   
Let $\mathcal{M}\in \mathbb{R}^{|\mathcal{A}|}$ be a menu. Given a menu $\mathcal{M}$ and a subset of alternatives $A\subseteq \mathcal{A}$, we denote  with $\mathcal{M}^{A}$ the \textquote{restricted} version of $\mathcal{M}$, with prices only for the alternatives in $A$. We say that a menu $\mathcal{M}$ is truthful for $v_{-i}$ if player $i$ never increases his utility by misreporting his type given the prices in $\mathcal{M}$.  
\begin{lemma} \label{uniqueconstraints1}
	Fix a type $v_{-i}$, denote its reachable alternatives as $A$ and its zero alternative as $a_0$. Let $\mathcal{M}$ be a normalized menu vector, i.e., $\mathcal{M}(a_{0})=0$. Then, $\mathcal{M}(a)-\mathcal{M}(b) \le \delta_{ab}(v_{-i})$ for every $a,b\in A$ if and only if $\mathcal{M}^{A}=\mathcal{M}^{A}_{v_{-i}}$. Moreover, every truthful menu $\mathcal{M}$ for $v_{-i}$ satisfies that for all reachable alternatives $a,b\in A$,  $\mathcal{M}(a)-\mathcal{M}(b)\le \delta_{ab}(v_{-i})$. 
\end{lemma}
\begin{proof} We first show that $\mathcal{M}_{v_{-i}}$ satisfies the constraints. Let $a,b$ be two reachable alternatives from $v_{-i}$. Recall that mechanism $M$ implements $f$, so for all $v$ such that $f(v,v_{-i})=a$:
	\begin{equation*}
	v(a)-\mathcal{M}_{v_{-i}}(a) \ge  v(b)-\mathcal{M}_{v_{-i}}(b)
	\implies v(a)-v(b) \ge  \mathcal{M}_{v_{-i}}(a) -\mathcal{M}_{v_{-i}}(b) 
	\end{equation*}
	This weak inequality holds for all the elements in $\{v(a)-v(b)\hspace{0.25em}|\hspace{0.25em}f(v,v_{-i})=a\}$, so it also holds for the infimum: 
	\begin{equation} \label{truthfuleq1}
	\delta_{ab}(v_{-i})=\inf\{v(a)-v(b)\hspace{0.25em}|\hspace{0.25em}f(v,v_{-i})=a\}\ge \mathcal{M}_{v_{-i}}(a) -\mathcal{M}_{v_{-i}}(b)
	\end{equation}
	Clearly, the same argument implies that every truthful menu $\mathcal{M}$ for $v_{-i}$ satisfies (\ref{truthfuleq1}).
	  
	Now, let $\mathcal{M}^{A}\neq \mathcal{M}^{A}_{v_{-i}}$ be a normalized menu vector that satisfies $\mathcal{M}(a)-\mathcal{M}(b) \le \delta_{ab}(v_{-i})$ for all $a,b\in A$. It means that $f(v,v_{-i})=a$ implies that: 
	\begin{eqnarray*}
		v(a)-v(b)\ge \inf\{v(a)-v(b)\hspace{0.25em}|\hspace{0.25em}f(v,v_{-i})=a\}=\delta_{ab}(v_{-i})\ge \mathcal{M}(a)-\mathcal{M}(b)\implies \\ v(a)-v(b)\ge \mathcal{M}(a)-\mathcal{M}(b)\implies v(a)-\mathcal{M}(a)\ge v(b) - \mathcal{M}(b)
	\end{eqnarray*}
	In words, the menu $\mathcal{M}^{A}$ induces truthful behaviour for player $i$ given $v_{-i}$. We extend $\mathcal{M}^{A}$ with the same prices as $\mathcal{M}^{A}_{v_{-i}}$ for all $a\notin A$ and obtain a non-restricted menu, $\mathcal{M}$.
	By assumption, $\mathcal{M}$ 
	is  normalized, so we get that it is truthful and normalized, but differs from $\mathcal{M}_{v_{-i}}$, which is a contradiction to the uniqueness of payments assumption. 
\end{proof}
Denote the elements in $\Ima \delta_{ab}$ with $\{\delta_{1}<\ldots<\delta_{j}<\ldots\}$.
The fact that $\Ima \delta_{ab}$ is countable allows us to enumerate its elements. 
\begin{lemma} \label{typeexistslemma}
	Fix a type $v_{-i}$ and a pair of reachable alternatives $a,b\in \mathcal{A}$ from $v_{-i}$. Fix $\delta_j\in \Ima \delta_{ab}$ where $j<|\Ima \delta_{ab}|$. Then, $\delta_{ab}(v_{{-i}})\le \delta_j$ if and only if there exists a type $v\in V_i$ such that:
	\begin{enumerate}
		\item \label{v1} $f(v,v_{-i})=a$.
		\item \label{v2}  $v(a)-v(b)< \delta_{j+1}$.
	\end{enumerate}
	Therefore, the existence of a type  that satisfies conditions \ref{v1} and \ref{v2} implies that every truthful menu $\mathcal{M}$ for $v_{-i}$ satisfies that $\mathcal{M}(a)-\mathcal{M}(b)\le \delta_j$.  
\end{lemma}
\begin{proof} 
	Assume $\delta_{ab}(v_{-i}) \le \delta_j$. $\delta_j<\delta_{j+1}$, so $\delta_{ab}(v_{-i})<\delta_{j+1}$. Recall that by definition, $\delta_{ab}(v_{-i})=\inf\{v(a)-v(b)\hspace{0.25em}|\hspace{0.25em}f(v,v_{-i})=a\}$. The fact that $\delta_{ab}(v_{-i})$ is the greatest lower bound and not $\delta_{j+1}$ implies that it is not a lower bound at all, and thus there exists a type $v\in V_i$ such that $f(v,v_{-i})=a$ and $v(a)-v(b)< \delta_{j+1}$.     
	
	For the other direction, assume that there exists $v\in V_{i}$ such that  $f(v,v_{-i})=a$ and $v(a)-v(b)< \delta_{j+1}$. It means that $v(a)-v(b)\in \{v(a)-v(b)\hspace{0.25em}|\hspace{0.25em}f(v,v_{-i})=a\}$, so a lower bound for this subset is necessarily smaller than $v(a)-v(b)$, which is strictly smaller than $\delta_{j+1}$.  Thus, $\delta_{ab}(v_{-i})<\delta_{j+1}$. Combining it with the fact that  $\delta_{ab}(v_{-i}) \in \Ima \delta_{ab}$ yields that $\delta_{ab}(v_{-i})$ belongs in  $\{\delta_1,\ldots,\delta_j\}$ and thus $\delta_{ab}(v_{-i}) \le \delta_j$. Recall that by Lemma \ref{uniqueconstraints1} every truthful menu $\mathcal{M}$ for $v_{-i}$ satisfies that $\mathcal{M}(a)-\mathcal{M}(b)\le \delta_{ab}(v_{-i})$. Combining the two inequalities gives that every truthful menu  $\mathcal{M}$ of $v_{-i}$ satisfies $\mathcal{M}(a)-\mathcal{M}(b)\le \delta_{j}$. 
\end{proof}
\begin{proof}[Proof of Proposition \ref{asquaredtypeslemma}] 
	Fix a type $v_{-i}$ and an  alternative $a^{\ast}$ which is reachable from it. Denote its set of reachable alternatives with $A$, its zero alternative with $a_0$ and its normalized menu with $\mathcal{M}_{v_{-i}}$. 
	
	We construct a subset $R\subseteq V_i$ as follows. First, we add the zero type $v_0$ to $R$. Then, for every reachable alternative $a\in A$, we add to $R$ a type $v_i$ such that $f(v_i,v_{-i})=a$. For every \emph{ordered} pair   $(a,b) \in A$, where $\delta_{ab}(v_{-i})$ is not the largest element in $\Ima \delta_{ab}$, we add to $R$ the following type: denote with $j_{ab}$ the index of $\delta_{ab}(v_{-i})$ in $\Ima \delta_{ab}$.  By Lemma \ref{typeexistslemma}, there exists a type $v_{ab}$ such that $f(v_{ab},v_{-i})=a$ and $v_{ab}(a)-v_{ab}(b)<\delta_{j_{ab}+1}$. We add it to $R$. Notice that $|R|=\mathcal{O}(|\mathcal{A}|^2)$.   
	
	Let $v_{-i}'\in V_{-i}$ be a type such that $f(v,v_{-i})=f(v',v_{-i})$ for all $v\in R$. We want to prove that $price_{i}^{a^\ast}(v_{-i})=price_{i}^{a^\ast}(v_{-i}')$. 
	Denote the set of reachable alternatives of $v_{-i}'$ as $A'$.
	By construction, all alternatives in $A$ are  reachable from $v_{-i}'$ as well, so $A\subseteq A'$. Also, $f(v_0,v_{-i}')=f(v_{0},v_{-i}')=a_0$, i.e. $v_{-i}$ and $v_{-i}'$ have the same zero alternative, so a normalized menu for $v_{-i}'$ necessarily satisfies $\mathcal{M}(a_0)=0$. Fix an ordered pair $(a,b)\in A$, and denote $\delta_{ab}(v_{-i})$ with $\delta_{j_{ab}}$. If $j_{ab}<|\Ima \delta_{ab}|$, by construction there exists $v_{ab}\in R$ such that $a=f(v_{ab},v_{-i})=f(v_{ab},v_{-i}')$ and $v_{ab}(a)- v_{ab}(b)< \delta_{j_{ab}+1}$. Recall that $A\subseteq A'$, so $a$ and $b$ are reachable from $v_{-i}'$ as well. Hence, we can use Lemma \ref{typeexistslemma} for $v_{-i}'$ and get that  $\delta_{ab}(v_{-i}')\le \delta_{j_{ab}}$ and that every truthful menu $\mathcal{M}$ for $v_{-i}'$ satisfies that $\mathcal{M}(a)- \mathcal{M}(b) \le  \delta_{j_{ab}}$. If $j_{ab}=|\Ima \delta_{ab}|$, $\Ima \delta_{ab}$ is necessarily finite. By Lemma \ref{uniqueconstraints1}, the fact that alternatives $a$ and $b$ are reachable from $v_{-i}'$ implies that every truthful menu for $v_{-i}'$ satisfies:
	\begin{equation*}
	\mathcal{M}(a)- \mathcal{M}(b) \le \delta_{ab}(v_{-i}') \le \max_{v_{-i}\in V_{-i}}\delta_{ab}(v_{-i})=\delta_{j_{ab}}   
	\end{equation*}
	Denote with $\mathcal{M}^{A}_{v_{-i}}$ and with $\mathcal{M}^{A}_{v_{-i}'}$ the menus presented by the mechanism $M$ for $v_{-i}$ and for $v_{-i}'$, restricted to the alternatives in $A$. By Lemma \ref{uniqueconstraints1}, $\mathcal{M}^{A}_{v_{-i}}$ is the \emph{only} menu that satisfies 
	$\mathcal{M}(a_0)=0$ and $\mathcal{M}(a)-\mathcal{M}(b)\le\delta_{ab}(v_{-i})= \delta_{j_{ab}}$ for every $a,b\in A$.  As we have just shown, a truthful and normalized menu $\mathcal{M}$ for $v_{-i}'$ must satisfy those conditions as well. Therefore, $\mathcal{M}^A_{v_{-i}}\equiv \mathcal{M}^A_{v_{-i}'}$, because otherwise we get a contradiction to the uniqueness of $\mathcal{M}_{v_{-i}}^A$.  By definition, $a^\ast \in A$, so $\mathcal{M}^A_{v_{-i}}(a^\ast)=\mathcal{M}^A_{v_{-i}'}(a^\ast)$, and thus $price_{i}^{a^\ast}(v_{-i})=price_{i}^{a^\ast}(v_{-i}')$, which completes the proof. 
\end{proof}
Observe the following naive proof system. Denote the most efficient protocol of $f$ with $\Pi^{f}$. Fix a type $v_{-i}$. Notice that by Lemma \ref{Rabccfbound}, $\Ima \delta_{ab}$ is finite for all $a,b\in \mathcal{A}$ and $f$ satisfies uniqueness of payments, so there exists a set $R\subseteq V_i$ of  $\mathcal{O}(|\mathcal{A}|^2)$ types that satisfies the conditions of Proposition \ref{asquaredtypeslemma}. 
The protocol is as follows: 
the prover sends all the types in $R$, and the players simulate $\Pi^{f}(r,v_{-i})$ for every $r\in R$.  By Proposition \ref{asquaredtypeslemma}, it suffices for the players for the extraction of the price of $a^{\ast}$. 
However, this naive protocol might be too costly, because if the size of the domain $V_i$ is large, pointing to a single index in it might require too many bits.

We overcome this problem as follows.  Instead of sending the types in $R$ themselves, the prover sends for each $r\in R$ the leaf in $\Pi^{f}$ that $(r,v_{-i})$ reaches. By that, we take advantage of the facts that the protocol $\Pi^{f}$ is public and has at most $2^{cc(f)}$ leaves. 
We denote the leaf (combinatorial rectangle) of each $r\in R$ with $L^r=L^r_1\times\cdots\times L^r_n$, and the set of leaves sent by the prover with $\mathcal{L}$. We also define a set of types in $V_{-i}$ which are congruent with $\mathcal{L}$, $cands(\mathcal{L})=\{v_{-i}\in V_{-i} \hspace{0.25em} | \hspace{0.25em} \forall L\in \mathcal{L},v_{-i}\in L_{-i}  \}$. Before outputting a price, the players verify that:
\begin{enumerate}
	\item \label{rectanglecond} For every player $j\in N/\{i\}$, and for every $L\in \mathcal{L}$, $v_j\in L_j$. In other words, the leaves sent by the prover are congruent with the players' types.  
	\item \label{emptycond}  $cands(\mathcal{L})\neq \varnothing$. 
	\item \label{zerocond} There exists a leaf $L_0$ such that $(v_{0},v_{-i})\in L_0$. The alternative associated with this leaf is the zero alternative.
	\item \label{promcond3} There exists a leaf  $L\in \mathcal{L}$ labelled with alternative $a^{\ast}$. 
	\item \label{propcond} 	For  every $v_{-i}^1, v_{-i}^2\in cands(\mathcal{L})$,
	$price_{i}^{a^{\ast}}(v_{-i}^1)=price_{i}^{a^{\ast}}(v_{-i}^2)$. Denote this price as $price_{i}^{a^\ast}(\mathcal{L})$.    
\end{enumerate}
Note that verifying all conditions does not require any communication between the players. If the verification fails, they reject. Otherwise, they output $price_i^{a^\ast}(\mathcal{L})$. 
\paragraph{Correctness.} First, we will show that for every $v_{-i}$ and every reachable $a^{\ast}$, if the prover sends the set $R$ specified in Proposition \ref{asquaredtypeslemma}, the players output $price_i^{a^\ast}(v_{-i})$.
By construction, the set of leaves $\mathcal{L}$ sent by the prover satisfies that $v_{-i}\in L_{-i}$ for all $L_r\in \mathcal{L}$ because $\Pi^f(r,v_{-i})$ reaches the leaf $L_r$, so conditions \ref{rectanglecond} and \ref{emptycond} hold. Proposition \ref{asquaredtypeslemma} guarantees that conditions \ref{zerocond} and \ref{promcond3} hold as well. 

For condition \ref{propcond}, fix a type $v_{-i}' \in cand(\mathcal{L})$. 
We will show that for all  $r\in R$, $(r,v_{-i})$ and $(r,v_{-i}')$ belong in the same leaf of $\Pi^f$, and thus $f(r,v_{-i})=f(r,v_{-i}')$. Using Proposition \ref{asquaredtypeslemma}, it implies that  $price_i^{a^\ast}(v_{-i}')=price_i^{a^\ast}(v_{-i})$ as needed.       

To this end, fix a  type $r\in R$, and denote the leaf sent for it in $\mathcal{L}$ with $L^r$. By construction, $(r,v_{-i})$ reaches $L^r$ and $v_{-i}'\in L^{r}_{-i}$ because $v_{-i}' \in cand(\mathcal{L})$.  $L^r$ is a combinatorial rectangle, so using its mixing property, we get that $(r,v_{-i}')$ reaches $L^r$ as well.  

We still need to prove that if $a^{\ast}$ is not reachable from $v_{-i}$, i.e., $v_{-i}$ violates the promise, then the players reject all witnesses for it. We also need to prove that if $price_i^{a^\ast}(v_{-i})=p$, there is no witness that convinces the players that the price is different. 

 If $a^{\ast}$ is not reachable from $v_{-i}$, no witness for $v_{-i}$  satisfies conditions \ref{rectanglecond} and \ref{promcond3} simultaneously, so the players always reject types in $V_{-i}$ that violate the promise. Fix a type $v_{-i}$ such that $price_i^{a^\ast}(v_{-i})=p$, and denote its set of reachable alternatives with $A$. $f$ satisfies uniqueness of payments, so $\mathcal{M}_{v_{-i}}^{A}$ is the only restricted menu which is truthful and normalized for $A$.  By definition, $\mathcal{M}_{v_{-i}}^{A}(a^{\ast})=p$. If the players return a payment other than $p$, it means that there exist leaves in the protocol $\Pi^{f}$ that point at outcomes of $f$ that disqualify $p$ from being the price of $a^{\ast}$ for $v_{-i}$ because of violations of truthfulness or normalization. If $p$ is disqualified, the menu  $\mathcal{M}_{v_{-i}}^{A}$ is invalidated as well. However,  $\mathcal{M}_{v_{-i}}^{A}$ is truthful and normalized, so it  cannot be invalidated and we get a contradiction.
\paragraph{Communication.} The total communication of the protocol is $\mathcal{O}(|\mathcal{A}|^2\cdot cc(f))$ bits, because $R$ is of size $\mathcal{O}(|\mathcal{A}|^2)$ and for every type in $R$ the prover sends an index of a leaf in $\Pi^f$, using $cc(f)$ bits. 
\end{proof}
\subsection{Scalable and Convex Domains} 
Scalable domains satisfy unique payments (Proposition \ref{scalableuniquenessthm}) and convex domains are basically equivalent to scalable domains, so we get the following claims for free.   
 \begin{claim}  \label{scalabledeterministicthm}
	Let $f\colon V_1\times \cdots\times V_n\to \mathcal{A}$ be an implementable social choice function with scalable domains. Then, $cc_{IC}(f)\le poly(n,cc(f),|\mathcal{A}|)$. 
\end{claim}
\begin{proof}
	By Proposition \ref{scalableuniquenessthm}, $f$ satisfies uniqueness of payments. 
	Hence, by Theorem \ref{uniquepaymentsalgthm},  $cc_{IC}(f) \allowbreak \le \allowbreak poly(n,|\mathcal{A}|,cc(f))$.
\end{proof}
\begin{claim} \label{convexdetclaim}
	Let $f\colon V_1\times ..\times V_n \to \mathcal{A}$ be a social choice function with convex domains. Then, $cc_{IC}(f)\le poly(n,cc(f),|\mathcal{A}|)$.
\end{claim}
\begin{proof}
	By Lemma \ref{convextranslatesscalablelemma}, $(f,V)$ has a translation $(f^t,V^t)$ where the domains $V^t$ are scalable. 
	\begin{equation*}
	cc_{IC}(f)\underbrace{=}_\text{By Lemma \ref{translationsamecommcomplexitylemma}} cc_{IC}(f^t) \underbrace{\le}_\text{by Claim \ref{scalabledeterministicthm}} poly(n,|\mathcal{A}|,cc(f^t)) \underbrace{=}_\text{By Lemma \ref{translationsamecommcomplexitylemma}} poly(n,|\mathcal{A}|,cc(f))
	\end{equation*}
\end{proof}
%\begin{openquestion}
%	Is $cc_{IC}(f)\le poly(cc(f),{|\mathcal{A}|})$ or $cc_{IC}(f)\le poly(2^{|\mathcal{A}|}, cc(f))$ or $cc_{IC}(f)\le poly(2^{2^{|\mathcal{A}|}}, cc(f))$ for \emph{all} functions?
%\end{openquestion}

\section{Hardness of Computing Payments in a Menu}\label{sec-hardness}

In most of the paper we have assumed that we are given an implementable social choice function $f$, an instance $(v_1,\ldots, v_n)$ and our goal was to compute the payment of each player. In this section we explore a very related but slightly different setting. In this setting, we are given the mechanism $M$ (which is stronger than having $f$), but now we are only given the valuations $v_{-i}$ of all players except player $i$. We ask the following fundamental questions: for a given alternative $a$, is there some $v_i$ such that $f(v_i,v_{-i})=a$? If so, can we efficiently find it and compute its price?

Formally, let $M=(f,P):V_1\times \ldots\times V_n\to \mathcal{A}\times \mathbb{R}^n$ be a mechanism. 
	For every player $i$ and alternative $a\in \mathcal{A}$, we define $reach_i^a:V_{-i}\to \{0,1\}$ as a  function that indicates whether $a$ is reachable from $v_{-i}$ or not. We also define a search variant of $reach$,  $reachWitness_i^a:V_{-i}\to V_i\cup \{\perp\}$. $reachWitness_i^a(v_{-i})$ outputs a type $v_i$ such that $f(v_i,v_{-i})=a$, or $\perp$ if $a$ is not reachable from $v_{-i}$. The third function we define is $price_i^a:V_{-i} \to \mathbb{R} \cup \{\infty\}$ that returns the price of alternative $a$ presented by $v_{-i}$. Throughout this section, we assume that $M$ assigns non-reachable alternatives a $\infty$ price. 
%	It is an abuse of notation, because in  Section \ref{uniquesec}, $price_i^a$ is defined with a reachability promise, in 
%	contrast to this definition.  
The deterministic communication complexity of reachability of a mechanisms, denoted with ${cc(reach(f))}$, is defined to be ${\max_{i,a} \{cc(reach_i^a)\}}$. Similarly, we write $cc(reachWitness(f))$ for
${\max_{i,a}   \{cc(reachWitness_i^a)\}}$ and $cc(price(P))$ for $\max_{i,a} \{cc(price_i^a)\}$.

\subsection{Reachability is Hard} \label{reachabilityhardsection}
%We show a mechanism $M=(f,P)$ where computing $reach(f)$ is exponentially harder than computing $M$. 
\begin{proposition} \label{reachahardprop}
	There exists a mechanism $M=(f_k,P)$ with single parameter domain, three players and a constant number of alternatives such that $cc(reach(f_k))\ge \exp(cc(M))$. 
\end{proposition}
\begin{proof}
	Fix some integer $k$.
The mechanism $M=(f_k,P)$ is as follows. There are three players: Alice, Bob and Charlie. The set of alternatives is $\mathcal{A}=\{b,c,bc,n\}$. The domains of Alice and Bob are single parameter: $r_A\in \{0,\dots,k-1\}$ and $r_B, r_C\in\{0,1\}^k$. We define a function that converts the types of Bob and Charlie to integers, $int:\{0,1\}^k\to \mathbb{N}$.
Alice's value for all alternatives is identical and equal to her private information $r_A$, i.e., $w_A\equiv 1$. Thus, we use $v_A$ and $r_A$ interchangeably.  Bob's value for alternative $a$ is $int(r_B)\cdot w_B(a)$, where $w_B(b)=w_B(bc)=1$ and $w_B(c)=w_B(n)=0$. Similarly, 
Charlie's value for alternative $a$ is $int(r_C)\cdot w_C(a)$, where $w_C(c)=w_C(bc)=1$ and $w_C(b)=w_C(n)=0$.
%  are the integer representations of  $v_B\cdot w_B(a)$ and $v_C\cdot w_C(a)$, respectively ($v_B$ and $v_C$ in this case are referred to as integers in $\{0,\ldots,2^k-1\}$). Bob's public function is $w_B(b)=w_B(bc)=1$ and $w_B(c)=w_B(n)=0$. Similarly, Charlie' public function is $w_C(c)=w_C(bc)=1$ and $w_C(b)=w_C(n)=0$. 
Denote with $r_B(j)$ and $r_C(j)$ the $j$'th bits of the binary string associated with $v_B$ and with $v_C$.   
The social choice function $f_k:V_A\times V_B\times V_C\to \mathcal{A}$ is as follows. For every $v_A\in V_A, v_B\in V_B, v_C\in V_C$: 
 	\begin{equation*}
 f_k(v_A,v_B,v_C)=
 \begin{cases}
  bc   &\qquad 
 r_{v_B}(v_A) = r_{v_C}(v_A ) = 1 \\
 b &\qquad r_{v_B}(v_A) =0, r_{v_C}(v_A)= 1 \\
c &\qquad r_{v_B}(v_A) =1, r_{v_C}(v_A)= 0 \\
n &\qquad \text{otherwise.}
 \end{cases} 
 \end{equation*} 

In words, Bob gets an alternative he values more if Charlie's $v_A$'th bit is on, and vice versa. Therefore, for every player $i$, $w_i$ is either constant or depends entirely on the types of other players. Thus, a constant zero payment implements $f$. Observe that $cc(M)=\mathcal{O}(\log k)$, because $M$ is computed by a trivial protocol where Alice sends $v_A$ using $\log k$ bits, and then Bob and Charlie send a single bit each:  $r_B(v_A)$ and $r_C(v_A)$. 

In order to show a gap, consider $reach^{bc}_{Alice}:V_B\times V_C\to \{0,1\}$. Observe that alternative $bc$ is reachable for Alice if and only if $r_{v_B}$ and $r_{v_C}$ are not disjoint. Hence,  $reach^{bc}_{Alice}$ can be trivially reduced to $DISJ_k$, which requires $k$ bits of communication \cite{ccbook}. Thus, $cc(reach(f_k))\ge k$.  
\end{proof}	
\subsection{Hardness of Reachability Determines Hardness of Computing Payments} \label{priceisreachsection}
\cite{dobzinski2016computational} shows that every truthful mechanism $M=(f,P)$ for domain with additive valuations satisfies that $cc(price(P))\le cc(M)$. We add that for player decisive functions, $cc(price(P))\le poly(n,cc(M))$. Furthermore, we demonstrate that the hardness of $price(\cdot)$ stems from the hardness of $reach(\cdot)$ (Proposition \ref{priceMreachMupperbound}). In contrast, we prove that we cannot derive a similar upper bound for $ReachWitness(\cdot)$ (Proposition \ref{gap2thm}).   
\begin{definition}(Player Decisiveness)
	We say that a function $f\colon V_1\times\cdots\times V_n\to \mathcal{A}$ is \emph{player decisive} if for every player $i$, alternative $a\in \mathcal{A}$ and type $v_{-i}\in V_{-i}$, there exists $v_i\in V_i$ such that $f(v_i,v_{-i})=a$. 
\end{definition}	
\begin{proposition} \label{priceMreachMupperbound}
	Let $M=(f,P)$ be a truthful mechanism. Then:
	\begin{equation*}
	cc(price(P))\le cc(reach(f))+poly(n,cc(M))
	\end{equation*} 
	As a corollary, if $f$ is player decisive, $cc(price(P))\le poly(n,cc(M))$.
\end{proposition}  
\begin{proof}
	Fix a player $i$ and an alternative $a\in \mathcal{A}$. 
	Fix a type $v_{-i}\in V_{-i}$.  
	First, the players execute the protocol of $reach_i^a(v_{-i})$. If alternative $a$ is not reachable from $v_{-i}$, they output $\infty$. Otherwise, they know that it is reachable.  
	By Theorem \ref{multipartynondtodreductionpromise}, \emph{deterministically} computing the price of a reachable alternative takes $poly(n,N(price_i^a))$, where $N(price_i^a)$ is the number of bits needed for the prover to present a witness for $price_i^a(v_{-i})=P_i$, with the guarantee that no witness is ever approved if $a$ is not reachable from $v_{-i}$. 

Observe the following naive proof system. Denote the protocol of the \emph{mechanism} $M$ with $\pi$, and observe that is has at most $2^{cc(M)}$ leaves. A prover sends to the players a name of a leaf $L$ such that $L$ is labelled with alternative $a$ and $v_{-i}\in L_{-i}$. It takes $CC(M)$ bits. Each player $j\neq i$ accepts if $v_j\in L_j$. Notice that such a leaf exists if and only if $a$ is reachable from $v_{-i}$, so all witnesses for inputs that violate the promise are indeed rejected. Since the leaf belongs to a protocol of the \emph{mechanism}, it is labelled with the payment of player $i$ $P_i$, so the leaf serves as a proof that  $price_i^a(v_{-i})=P_i$. Thus, $N(price_i^a)\le cc(M)$.  
	Therefore,  for every player $i$ and alternative $a\in \mathcal{A}$:
	\begin{align}
	cc(price_i^a) 
	&\le   cc(reach_i^a) + poly(n,N(price_i^a)) \nonumber & \text{(by Theorem \ref{multipartynondtodreductionpromise})} \\
	&\le 
	cc(reach_i^a) + poly(n,cc(M)) \nonumber \\
	&\implies cc(price(P))\le cc(reach(f)) + poly(n,cc(M)) \nonumber
	\end{align}
	An observant reader might wonder whether this protocol can be used to derive an upper bound for $reachAndWitness(f)$ as well. Note that an $a-$leaf that $v_{-i}$ belongs to contains information not only about price of player $i$, but also about types in  $V_i$ that reach $a$ given $v_{-i}$.  However, this line of thought fails. The reason for it is that Theorem \ref{multipartynondtodreductionpromise} applies to functions, not to relations. Given $v_{-i}$, the mechanism presents a \emph{single} price for alternative $a$, so $price_i^a(v_{-i})$ is a function. In contrast, $v_{-i}$ has numerous types in $V_i$ that reach $a$. 
\end{proof}
Is it possible that $cc(reachWitness(f)) \allowbreak \le \allowbreak cc(reach(f))+poly(n,cc(M))$, but our proof technique fails to show it?  We answer this question in the negative. 
\begin{proposition} \label{gap2thm}
	Let $k$ be a number divisible by $3$. Then, there exists a mechanism $M=(f_k,P)$ with single parameter domain, three players and a constant number of alternatives such that $cc(reach(f_k))=cc(M)=\mathcal{O}(\log k)$, but  $cc(reachWitness(f_k))=\Omega(k)$.
\end{proposition}
We use the following relation in the proof of Proposition \ref{gap2thm}. 
\begin{definition} \label{matchdef}
Set $k=3m$. Let $X$ be the set of all graphs of $k$ vertices with a matching of size $m$, and let $Y$ be the set of all graphs of $k$ vertices with no matching of size $m$.  Define the relation $MATCH \subseteq X\times Y \times \{1,\ldots,\binom{k}{2}\}$ as following. $(x,y,i)\in MATCH$ if edge $i$ is in the graph $x$, but not in the graph $y$. 
\end{definition} 
Recall that a protocol computes a relation if for every $(x,y)\in X\times Y$, the protocol outputs $i$ such that $(x,y,i)\in R$.  
\begin{theorem}[\cite{ccbook,raz1992monotone}] \label{razmonotonematchthm}
	$cc(MATCH) = \Omega(k)$. 
\end{theorem}
	\begin{proof}[Proof of Proposition \ref{gap2thm}]
		The mechanism is as follows. There are three players: Bob, Charlie and Diane with single-parameter domains. The set of alternatives is $\mathcal{A}=\{b,c,bc,n\}$. 
		 Bob's type space is $X$ and Charlie's is $Y$, where $X$ and $Y$ are defined in Definition \ref{matchdef}. 
		 We define arbitrary mappings  $int_X:X\to \mathbb{N}$ and $int_Y:Y\to \mathbb{N}$ that convert  Bob's and Charlie's graph types to some integers. Bob's value for an alternative $a$ is $int_X(x)\cdot w_B(a)$, where $w_B(b)=w_B(bc)=1$ and $w_B(c)=w_B(n)=0$. Similarly, 
		 Charlie's value for $a$ is $int_Y(y)\cdot w_C(a)$, where  $w_C(c)=w_C(bc)=1$ and $w_C(b)=w_C(n)=0$.
		 We define $V_X$ and $V_Y$ to be the valuations sets of Bob and Charlie, respectively. 
		 Diane's valuations set is $\{1,\ldots,\binom{k}{2}+3\}$, i.e.
		 her 
		 value for all alternatives is identical and equal to her private information ($w_D\equiv 1$). Recall that $w_B,w_C,w_D:\mathcal{A}\to \mathbb{R}$ are public. 
		  For every valuation $v_x\in V_x$ and every $1\le j \le \binom{k}{2}$, let $v_{x}(j)$ be $1$ if the $j$'th edge exists in the graph $x$ associated with the valuation $v_x$, and  define $v_{y}(j)$ similarly. $f_k\colon V_X\times V_Y \times \{1,\ldots,\binom{k}{2}\} \to \mathcal{A}$ is as follows. For every $v_x\in V_X$, $v_y\in V_Y$ and $e\in \{1,\ldots,\binom{k}{2}+3\}$:
		  \begin{equation*}
		 f_k(v_x,v_y,e)=
		 \begin{cases}
		 bc   &\qquad e=\binom{k}{2}+1 
		 \quad
		 \text{or} \quad (e \le \binom{k}{2}\hspace{0.4em} \text{and} \hspace{0.4em} v_x(e) = v_y(e) = 1 ) 
		 \\
		 n &\qquad e=\binom{k}{2}+2 
		 \quad
		 \text{or} \quad (e \le \binom{k}{2}\hspace{0.4em} \text{and} \hspace{0.4em} v_x(e) = v_y(e) = 0 )  \\
		 b &\qquad 
		 e=\binom{k}{2}+3
		 \quad
		 \text{or} \quad (e \le \binom{k}{2}\hspace{0.4em} \text{and} \hspace{0.4em} v_x(e)=0 \hspace{0.4em} \text{and} \hspace{0.4em} v_y(e) = 1 ) 
		 \\
		 		 c &\qquad v_x(e)=1 \hspace{0.4em} \text{and} \hspace{0.4em} v_y(e) = 0 
		 \end{cases} 
		 \end{equation*}  
		 Intuitively, Diane has three \textquote{special} valuations that make the alternative chosen be one of  $\{bc,n,b\}$, regardless of the other players' valuations. If she does not hold one of them, then if $e\in x$ (i.e. the edge $e$ belongs in  Bob's graph), an alternative that benefits Charlie is chosen, and vice versa. 
		 
		 Due to those special types,  $\{bc,n,b\}$ are reachable from $(v_x,v_y)$ for all $(v_x,v_y)\in V_x\times V_y$. Alternative $c$ is reachable from all $(v_x,v_y)$ as well: by construction $x$ has an $m-$matching, and $y$ does not, so there necessarily exists an edge that belongs in $x$ but not in $y$. Therefore, $cc(reach_{Diane}^a)=0$ for all $a\in \mathcal{A}$. Consider $reach_{Bob}^a:V_Y\times \{1,\ldots,\binom{k}{2}\}\to \{0,1\}$ for some $a\in \mathcal{A}$. Diane can send her type $e$ using $\mathcal{O}(\log k)$ bits, and then Charlie sends back whether $y$ contains the edge $e$ or not. If it contains $e$, only alternatives $\{b,bc\}$ are reachable, otherwise only alternatives $\{c,n\}$ are reachable. Therefore, $cc(reach_{Bob}^a)=\mathcal{O}(\log k)$ for every $a\in \mathcal{A}$. Similarly, $cc(reach_{Charlie}^a)=\mathcal{O}(\log k)$ for every $a\in \mathcal{A}$. Therefore, $CC(reach(f_k))=\mathcal{O}(\log k)$.  
		 
		 Observe that none of the players benefit from misreporting, because for all of them $w_i$ is either constant or depends entirely on the valuations of the other players.  
		Thus, a mechanism $M$ with no payments at all implements $f$.  
Observe that $M$ has the following trivial protocol: Diane sends her valuation $e$, and Bob and Charlie reply with $v_x(e)$ and $v_y(e)$. Thus, all players know the alternatives chosen and that everyone pays zero. Therefore, $cc(M)=\mathcal{O}(\log k)$. 
		 
		 We now want to show that the computation of $reachWitness(f_k)$ is hard. Consider $reachWitness_{Diane}^c$, the function that outputs a valuation in $\{1,\ldots,\binom{k}{2}\}$ that reaches alternative $c$, given  $(v_x,v_y)\in V_X\times V_Y$. As we explained earlier, such a valuation always exists. Observe that $f(v_x,v_y,e)=c$ if and only if the edge $e$ belongs in $x$ and not in $y$. Therefore, computing a valuation $e$ that reaches alternative $c$ is equivalent to computing the relation $MATCH$. Thus, The lower bound on $MATCH$ (Theorem \ref{razmonotonematchthm}) implies   $cc(reachWitness(f_k))=\Omega(k)$. 
	\end{proof}
\subsection{Exponential Upper Bounds on Reachability and on  Payment Computation} \label{reachandkeysection}
We now show that for every mechanism $M=(f,P)$, $cc(reach(f))$, $cc(reachWitness(f))$ and $cc(price(P))$ are at most $\exp(cc(M))$. Thus, the lower bounds presented in Propositions \ref{reachahardprop} and \ref{gap2thm} are tight.
\begin{proposition}
	Let $M=(f,P)$ be a truthful mechanism. Then, $cc(reachWitness(f))\le \exp(cc(M))$. As a corollary, $cc(reach(f))\le \exp(cc(M))$ and $cc(price(P))\le \exp(cc(M))$.  
\end{proposition}
The proof is motivated by Fadel and Segal \cite{fs}. It relies on the inner structure of protocols, so we define them as follows. The definition is based on \cite{babaioff2015truthful}.
\begin{definition} (Communication Protocol)
A protocol $\pi$ is defined for $n$ players, type spaces $V=V_1\times \cdots \times V_n$ and an outcome space $O$. It is a binary tree with decision nodes $U$, and leaves $\mathcal{L}$. Each decision node $u\in U$ is associated with a player $i$ and with a decision function $s_{u}:V_i\to \{0,1\}$. The computation of a protocol is as follows. If $s_u(v_i)=1$, the players proceed to its right child, and otherwise to its left child. We say that a protocol computes a function $f$ if for all $v\in V$, $\pi(v)=f(v)$. 
\end{definition}
\begin{proof}
	First, we see that the proposition implies its corollaries. By definition, $cc(reach(f))\le CC(reachWitness(f))$, so an upper bound for $cc(reachWitness(f))$ applies to $cc(reach(f))$ as well. 
	Regarding $price(P)$, 
	 observe that for each player $i$ and alternatives $a\in \mathcal{A}$,
	$price_i^a$ can be computed by one execution of the protocol of $reachWitness_i^a$ and another execution of the protocol of $M$.
	The first execution yields a type $v_i$ such that $f(v,v_{-i})=a$ or a price of $\infty$, and the second execution simulates $M(v,v_{-i})$ and outputs the price. Thus:
	\begin{equation*}
	 cc(price(P))\le cc(reachWitness(f))+cc(M)\le \exp(cc(M))
	\end{equation*}
	
	Now, we prove that $cc(reachWitness(f))\le \exp(cc(M))$. Denote the optimal protocol of $M$ with $\pi$. 
	The players agree in advance on a profile $(v_1,\ldots,v_n)\in L= L_1\times \cdots\times L_n$ for every leaf in $\pi$. Fix a player $i$ and an alternative $a$. The computation of $reachWitness_i^a(v_{-i})$ is as follows. For every inner node  $u\in U$, the player $j$ who is in charge of it sends $s_u(v_j)$. If an inner node belongs to player $i$, the players skip it. There are at most $2^{cc(M)}$ decision nodes, so it takes $2^{cc(M)}$ bits.     
	
	The next phase of the protocol requires no communication, so all players perform it  simultaneously. Observe that each leaf is associated with a path in the binary tree, which might include decision nodes of player $i$.  Each player   
goes over each $L\in \mathcal{L}$ labelled with alternative $a$ in a an agreed upon order, and checks whether the bits of decision nodes sent by all the players are compatible with the path of the leaf. If there is a compatible $a$-leaf $L$, the players output his agreed upon type $v_i\in L_i$. If there is no such leaf, it means that $a$ is not reachable from $v_{-i}$, and the players output $\perp$. The total communication is $2^{cc(M)}$ bits.    
\end{proof}

\bibliographystyle{alpha}
\bibliography{sample}

\newcommand{\etalchar}[1]{$^{#1}$}
\begin{thebibliography}{HMUV09}

\bibitem[AKSW20]{assadi2020separating}
Sepehr Assadi, Hrishikesh Khandeparkar, Raghuvansh~R Saxena, and S~Matthew
  Weinberg.
\newblock Separating the communication complexity of truthful and non-truthful
  combinatorial auctions.
\newblock In {\em Proceedings of the 52nd Annual ACM SIGACT Symposium on Theory
  of Computing}, pages 1073--1085, 2020.

\bibitem[APTT04]{archer2004approximate}
Aaron Archer, Christos Papadimitriou, Kunal Talwar, and {\'E}va Tardos.
\newblock An approximate truthful mechanism for combinatorial auctions with
  single parameter agents.
\newblock {\em Internet Mathematics}, 1(2):129--150, 2004.

\bibitem[AT01]{archer2001truthful}
Aaron Archer and {\'E}va Tardos.
\newblock Truthful mechanisms for one-parameter agents.
\newblock In {\em Proceedings 42nd IEEE Symposium on Foundations of Computer
  Science}, pages 482--491. IEEE, 2001.

\bibitem[BBNS08]{schapira}
Moshe Babaioff, Liad Blumrosen, Moni Naor, and Michael Schapira.
\newblock Informational overhead of incentive compatibility.
\newblock In {\em Proceedings of the 9th ACM conference on Electronic
  commerce}, pages 88--97, 2008.

\bibitem[BCL{\etalchar{+}}06]{bikhchandani}
Sushil Bikhchandani, Shurojit Chatterji, Ron Lavi, Ahuva Mu'alem, Noam Nisan,
  and Arunava Sun.
\newblock Weak monotonicity characterizes deterministic dominant-strategy
  implementation.
\newblock {\em Econometrica}, 74(4):1109--1132, 2006.

\bibitem[BKS15]{babaioff2015truthful}
Moshe Babaioff, Robert~D Kleinberg, and Aleksandrs Slivkins.
\newblock Truthful mechanisms with implicit payment computation.
\newblock {\em Journal of the ACM (JACM)}, 62(2):1--37, 2015.

\bibitem[DF92]{dolev1992determinism}
Danny Dolev and Tom{\'a}s Feder.
\newblock Determinism vs. nondeterminism in multiparty communication
  complexity.
\newblock {\em SIAM Journal on Computing}, 21(5):889--895, 1992.

\bibitem[Dob16]{dobzinski2016computational}
Shahar Dobzinski.
\newblock Computational efficiency requires simple taxation.
\newblock In {\em 2016 IEEE 57th Annual Symposium on Foundations of Computer
  Science (FOCS)}, pages 209--218. IEEE, 2016.

\bibitem[FS09]{fs}
Ronald Fadel and Ilya Segal.
\newblock The communication cost of selfishness.
\newblock {\em Journal of Economic Theory}, 144(5):1895--1920, 2009.

\bibitem[HMUV09]{vohrarevneueequiv}
Birgit Heydenreich, Rudolf Müller, Marc Uetz, and Rakesh~V. Vohra.
\newblock Characterization of revenue equivalence.
\newblock {\em Econometrica}, 77(1):307--316, 2009.

\bibitem[KN96]{ccbook}
Eyal Kushilevitz and Noam Nisan.
\newblock {\em Communication Complexity}.
\newblock Cambridge University Press, 1996.

\bibitem[KNSW94]{karchmer1994non}
Mauricio Karchmer, Ilan Newman, Mike Saks, and Avi Wigderson.
\newblock Non-deterministic communication complexity with few witnesses.
\newblock {\em Journal of Computer and System Sciences}, 49(2):247--257, 1994.

\bibitem[Mye81]{myerson1981optimal}
Roger~B Myerson.
\newblock Optimal auction design.
\newblock {\em Mathematics of operations research}, 6(1):58--73, 1981.

\bibitem[Nis07]{noamchapter}
Noam Nisan.
\newblock {\em Introduction to Mechanism Design (for Computer Scientists)},
  page 209–242.
\newblock Cambridge University Press, 2007.

\bibitem[Roc87]{rochet1987necessary}
Jean-Charles Rochet.
\newblock A necessary and sufficient condition for rationalizability in a
  quasi-linear context.
\newblock {\em Journal of mathematical Economics}, 16(2):191--200, 1987.

\bibitem[RST{\etalchar{+}}20]{RSTWZ20}
Aviad Rubinstein, Raghuvansh~R. Saxena, Clayton Thomas, S.~Mathew Weinberg, and
  Junyao Zhao.
\newblock Exponential communication separations between notions of selfishness.
\newblock 2020.

\bibitem[RW92]{raz1992monotone}
Ran Raz and Avi Wigderson.
\newblock Monotone circuits for matching require linear depth.
\newblock {\em Journal of the ACM (JACM)}, 39(3):736--744, 1992.

\bibitem[SY05]{saksyuconvex}
Michael Saks and Lan Yu.
\newblock Weak monotonicity suffices for truthfulness on convex domains.
\newblock In {\em Proceedings of the 6th ACM Conference on Electronic
  Commerce}, EC '05, page 286–293, New York, NY, USA, 2005. Association for
  Computing Machinery.

\end{thebibliography}

\appendix

\section{Deterministic and Nondeterministic Communication Complexity} \label{promisesection}

The proofs of the unique payments algorithm (Theorem \ref{uniquepaymentsalgthm})  and the upper bound on $price(M)$ (Proposition \ref{priceMreachMupperbound})  rely on the polynomial relations between deterministic and non-deterministic communication complexity of promise problems. We hereby prove this assertion. 
We stress that despite the fact that we focus on promise problems, all the results of this section can be easily extended to non-promise problems. Usually, given a deterministic function $f\colon X=X_1\times \cdots \times X_n \to O$ and a protocol $\pi$ we say that  $\pi$ computes $f$ if $f(x)=\pi(x)$ for all $x\in X$. For promise problems, we allow $\pi$ to err sometimes. Formally:  
% A promise problem for a function $f$ prescribes a subset of promise inputs, $P\subseteq X$.  
% the players are guaranteed to get $x\in $ are some inputs where a protocol for $f$ is no
%The notion of promise problems is similar to that of c 
%Similarly to the notion of promise problems in computational complexity, we define promise  
\begin{definition} \label{promisedefinition} (Promise Problems)
Let $f:X=X_{1}\times \cdots\times X_{n}\to O \cup \{\ast \}$ be a function. We call $x$ such that $f(x)\in O$ \emph{promise inputs}, and denote them with $P\subseteq X$. We say that a \emph{deterministic} protocol $\pi$ computes $f$ if $\pi(x)=f(x)$ for all $x\in P$. We say that a non-deterministic protocol computes $f$ if it presents a valid witness for $x\in P$, whilst all witnesses are rejected if $x$ violates the promise. We denote with $N(f)$ the non-deterministic communication complexity of $f$.  
\end{definition}
Recall that a non-deterministic protocol is equivalent to a cover of inputs. Thus, we require that the cover associated with a non-deterministic protocol for $f$ satisfies that every $o$-monochromatic rectangle in it does not contain $\ast$-inputs, i.e. inputs which violate the promise. 
Our goal is to prove that:
\begin{theorem}
	\label{multipartynondtodreductionpromise}
	Let $f\colon X_{1}\times \cdots\times X_{n}\to O\cup \{\ast\}$ be a promise function.   Then, there exists a deterministic protocol that computes $f$ using $poly(n,N(f))$ bits. 
\end{theorem}
To this end, we define for each $o\in O$ a verifier function which is a promise problem as well:    $V_{f}^{o}:X\to\{0,1,\ast\}$. $V_f^o$ has the same promise inputs as the original $f$. It outputs $1$ if $f(x)=o$, $0$ if $f(x)\in O$ but differs from $o$ and $\ast$ if $x\notin P$. We define the most costly verifier with $V_f$, i.e. $V_{f}=\argmax_{o\in O}cc(V_{f}^{o})$.  
% \newline \newline  Denote the outcome whose verifier is the most costly one with $o^{}$ most costly one with $\pi^o$   
% Similarly to the definition of $f$, we are interested in  computing a protocol whose inputs are solely $x$'s such that $f(x)\in \{0,1\}$. 
%Denote with $D(v_f^o)$ the minimal number of bits required for the computation of such a protocol. 
%% only $0-$inputs  have a promise that we always get $x$ such that $f(x)\in \{0,1\}$. 
%Similarly, denote by $D(V_{f})$ the verifier whose computation is the most  costly, i.e. $D(V_{f})=\max_{o\in O}D(V_{f}^{o})$. 
We prove the theorem by combining the two following propositions: 
\begin{proposition}\label{computationisverification}
	Let $f:X_{1}\times\cdots\times X_{n}\to O\cup \{\ast\}$ be a promise problem. Then,  $cc(f)\le \mathcal{O}\Big(n^{2}\big(\log^{2} |O|+ \log|O|\cdot cc(V_{f})+ cc^{2}(V_{f})\big)\Big)$.
\end{proposition}
 \begin{proposition} \label{proptwoplayerpromise} 
	Let $f\colon X=X_{1}\times \cdots\times X_{n}\to \{0,1,\ast\}$ be a \emph{boolean} promise function.   Then, there exists a deterministic protocol that computes $f$ using $\mathcal O(n^2\cdot N^0(f)\cdot N^1(f))$ bits. Thus, $cc(f)\le poly(n,N(f))$. 
\end{proposition}
\begin{proof}[Proof of Theorem \ref{multipartynondtodreductionpromise}]
%	We now generalize the results of Theorem \ref{multipartynondtodreductionpromise} beyond the cases where $\mathcal{O}=\{0,1\}$. 
%	First, we denote with $V_f^o:X\to \{0,1,\ast\}$ a verifier for the outcome $o$, a function that outputs $1$ if $f(x)=o$,  $0$ if $f(x)$ is equal to some other $o\in O$. For promise inputs, it is let off the hook. Also, we denote $cc(V_{f})=\max_{o\in O}cc(V_{f}^{o})$ and $N(V_{f})=\max_{o\in O}N(V_{f}^{o})$.
%	
By Proposition  \ref{computationisverification}, we get that:
	\begin{equation*}
	cc(f)\le \mathcal{O}\Big(n^{2}\big(\log^{2} |O|+ \log|O|\cdot cc(V_{f})+ cc^{2}(V_{f})\big)\Big)
	\end{equation*}
	$N(f)\ge \log |O|$, so we only need to upper bound $cc(V_f)$. By Proposition \ref{proptwoplayerpromise}, $cc(V_f) \le poly(n, N(V_f))$. 
%	Any of the verifiers $V_f^{o}$ can be computed with an oracle access to $f$, by simply outputting $1$ if $f(x)=o$ and $0$ otherwise. 
	Clearly, $N(V_f)\le N(f)$ because a non-deterministic protocol for $f$ can be used to compute every $V_f^{o}$ because $f$ and $V_f^{o}$  share the same promise $P\subseteq X$. Therefore:
	\begin{align*}
	cc(f) &\le \mathcal{O}\Big(n^{2}\big(\log^{2} |O|+ \log|O|\cdot cc(V_{f})+ cc^{2}(V_{f})\big)\Big) \\ 
	&\le poly (n, N(f), cc(V_f)) \\ 
	&\le poly(n, N(f), N(V_f)) & \text{(by Proposition \ref{proptwoplayerpromise})} \\
	&\le poly(n, N(f))  
	\end{align*}
\end{proof}
\subsection{Reduction of Computation to Verification}
The proof is by a reduction to $unique\text{-}disjointness$. $unique\text{-}disjointness$ is a search problem, where each player holds $l$ bits, with the promise that \emph{at most} one of them is intersecting.  A bit is said to be  \emph{intersecting} if it is turned on for all players. A protocol solves the problem if it returns the index of the intersecting bit if it exists, and $\perp$ otherwise.\footnote{The definition of $unique\text{-}disjointness$ is similar to the one in \cite{dobzinski2016computational}, except that we redefine it as a search problem instead of a decision problem.} 
\begin{theorem} \label{zdisj}\cite{dobzinski2016computational,karchmer1994non}
	There is a protocol with communication complexity $\mathcal{O}(n^{2}\cdot \log^{2} l)$ which solves $unique\text{-}disjointness$.
\end{theorem}
\begin{proof}[Proof of Proposition \ref{computationisverification}]
	The reduction of $f$ to  $unique\text{-}disjointness$ is as follows. First, each player translates her input $x_{i}\in X_{i}$ to a new form $b_{i}$ as follows. Let $\pi^{o}$ be the most efficient communication protocol for $V_{f}^{o}$. Each player holds a bit in $b_{i}$ for every \emph{$1$}-leaf in $\pi^{o}$ for all $o \in O$. Denote the number of bits each player holds as $l$, and note that $l\le |O| \cdot  2^{cc(V_{f})}$, since there are $|O|$ protocols and  the number of leaves in each of them is at most the exponent of its height. Recall that a leaf is a combinatorial rectangle, i.e., it is of the form $\prod_{j=1}^{n}S_{j}$, where $S_{j}\subseteq X_{j}$. For player $i$, the bit that represents a 1-leaf in $b_{i}$ is turned on only if her input belongs in the leaf, i.e. $x_{i}\in S_{i}$. Hence, if the players hold $x\in X$, a bit of a 1-leaf in $\pi^{o}$ is intersecting if and only if the execution of $\pi^{o}(x)$ ends at this particular leaf. 
	
	Recall that we are promised to get $x$ such that $f(x)\in O$. Denote $f(x)=o^{\ast}$. It means that $V_{f}^{o^{\ast}}(x)=1$, so $\Pi^{o^{\ast}}(x)$ reaches a $1-$leaf. For all other outcomes $o\neq o^{\ast}$, by design the protocol $\Pi^{o}(x)$ reaches a 0-leaf because $x$ satisfies the promise, so none of the $1$-leaves of those protocols are intersecting. Hence, there is a \emph{single} $1$-leaf whose bit is intersecting, and it belongs to $\Pi^{o^{\ast}}$. Thus, the following protocol computes $f$: the players simulate the communication protocol for $unique\text{-}disjointness$ with $b_{1},\ldots,b_{n}$. Afterwards, they return the outcome whose protocol has $1-$leaf whose bit is intersecting. By Theorem \ref{zdisj}, it takes:
	\begin{equation*}
	\mathcal{O}(n^{2}\cdot \log^2 l )=\mathcal{O}(n^{2}\cdot \log^{2}(|O| \cdot  2^{cc(V_{f})}))=\mathcal{O}\Big(n^{2}\big(\log^{2} |O|+ \log|O|\cdot cc(V_{f})+ cc^{2}(V_{f})\big)\Big)
	\end{equation*}
\end{proof}	
\subsection{Deterministic and Nondeterministic Communication of Boolean Promise Problems}
%We begin by proving the theorem for $\mathcal{O}=\{0,1\}$, and then we leverage this result to non-boolean functions. 
We show that the proof in the two party model in \cite[Theorem~2.11]{ccbook} can be extended to a multi-player promise setting.  Dolev and Feder \cite{dolev1992determinism} provide a similar result, but they do not address the promise scenario.   
% \begin{proposition} \label{proptwoplayerpromise} 
%	Let $f\colon X=X_{1}\times \cdots\times X_{n}\to \{0,1,\ast\}$ be a promise function.   Then, there exists a deterministic protocol that computes $f$ using $\mathcal O(n^2\cdot N^0(f)\cdot N^1(f))$ bits.
%\end{proposition}
We begin by proving the following simple combinatorial fact.
	\begin{lemma} \label{combilemma}
		Let $R$ be a set of objects, and let $R_1,...,R_n\subseteq R$ be $n$ subsets such that $\bigcap_{j=1}^n R_j=\varnothing$. Then, there exists $j\in[n]$ such that $|R_j|\le \frac{n-1}{n}\cdot |R|$.  
	\end{lemma}
\begin{proof}[Proof of Lemma]
	For every $j$, denote with $\overline{R}_j=\{x\in S, x\notin R_j\}$. If for all $j$, $|R_j|> \frac{n-1}{n}\cdot |R|$, it means that there exists $\epsilon>0$ such that 
	$|R_j|> \frac{n-1}{n}\cdot |R|+\epsilon$, so by definition 
	 $|\overline{R}_j|\le \frac{1}{n}\cdot |R|-\epsilon$. By the union bound, $|\bigcup_{j=1}^n \overline{R}_j| < |R|$, so there exists $x$ such that $x\in R_j$ for all $j\in [n]$, a contradiction. 
\end{proof}
\begin{proof}[Proof of Proposition \ref{proptwoplayerpromise}]
We use the same protocol as the one in \cite[Theorem~2.11]{ccbook}, with minor adjustments. Let $\pi^0$ and $\pi^1$ be  non-deterministic protocols for the promise problem $f$, whose communication complexities are $N^{0}(f)$ and $N^{1}(f)$, respectively. Denote with $C^{0}$ and $C^{1}$ the covers associated with $\pi^0$ and with $\pi^1$, and recall that by definition all the rectangles in them do not contain $\ast-$inputs.  
%Recall that a cover is a set of monochromatic rectangles (with no $\ast$-inputs) such that each $x\in X$ such that $f(x)=1$ belongs to at least rectangle in $C^1$, and the same goes for $C^0$. 
For every $i\in[n]$, we say that two rectangles $R^1=R_1^1\times\cdots \times R_n^1$ and $R^2=R_1^2\times\cdots \times R_n^2$ \emph{$i-$intersect} if $R_i^1\cap R_i^2 \neq \varnothing$.  

In each iteration $t$, the players holds a subset of $0$-rectangles from $C^0(f)$. 
%that are \textquote{suspected} of having the shared type $(x_1,\ldots,x_n)$ in them. 
We call this list $live(t)$, and initialize it to be $C^0(f)$. The protocol is as follows:
in each iteration $t$, if $Live(t)=\varnothing$, the players output $1$. Otherwise, every player $i$ in his turn  
checks whether there exists a $1-$rectangle $R^t=R^t_1\times \cdots \times R^t_n$ in $C^0$ such that $R^t$ $i-$intersects at most $\frac{n-1}{n}$ of the rectangles in $live(t)$  and also $x_i\in R_i$. If a player found such a rectangle, he sends its name to the other players and they all update:
$live(t+1)\gets \{ \text{$R$ $i-$intersects $R^t$} \hspace{0.25em}| \hspace{0.25em} R\in live(t) \}$ and proceed to the next iteration. If none of the players found a rectangle which satisfies those conditions, the players output $0$. 

For correctness, denote a rectangle that $x$ belongs to with $R^\ast$.  If $x$ is a $0-$input, by construction it remains inside $live(t)$ in every iteration, so $live(t)$ never empties out and the players indeed output $0$. If $x$ is a $1-$input, we will show that in each iteration where $live(t)\neq \varnothing$, $R^\ast$ satisfies the conditions to be  for $R_t$ for at least one of the players. Fix an iteration $t$, and divide the set of $0-$rectangles $live(t)$ to $n$ subsets:
$live_1(t),\ldots,live_n(t)$ where $live_i(t)$ is the set of rectangles in $live(t)$ that $i-$intersect $R^\ast$. $R^\ast$ is a $1$-rectangle, so $\bigcap_{i=1}^n live_i(t)$ is necessarily empty, because if there exists a $0-$rectangle which $i-$intersects with $R^\ast$ for all $i\in [n]$, it means that they have a shared input, which is a contradiction.\footnote{This is why we do not allow monochromatic rectangles to contain promise inputs. If we had allowed them to contain promise inputs, a $0$-rectangle and $1$-rectangle could have been intersecting.} 
Since $\bigcap_{i=1}^n live_i(t)$ is empty, by Lemma \ref{combilemma}, there exists an index  $j$ such that $|live_j(t)|\le \frac{n-1}{n}\cdot  live(t)$, so player $j$ necessarily finds a rectangle.     

Regarding the communication complexity, observe that in each iteration at most $n+\log |C^1|$ bits are sent. Also, $live(t)$ decreases by a multiplicative factor of $\frac{n-1}{n}$ in each iteration, so there are at most $\mathcal{O}(n\cdot \log|C^0|)$ iterations. Thus, the total communication is $\mathcal{O}(n^2\cdot N^0(f)\cdot N^1(f))$. 
\end{proof}

\end{document}